\documentclass[a4paper,10pt]{amsart}
\usepackage[utf8]{inputenc}
\usepackage[english]{babel}

\usepackage{amsfonts}

\usepackage[english]{babel}
\usepackage{amsmath,amssymb,verbatim,mathrsfs,latexsym,paralist}
\usepackage{graphicx}
\usepackage{epstopdf}
\usepackage{indentfirst}
\usepackage{multirow}
\usepackage{xcolor}
\usepackage{pdflscape}
\usepackage{amsthm}
\usepackage{longtable}
\allowdisplaybreaks[4]

\newcommand{\Z}{\mathbb{Z}}
\newcommand{\C}{\mathbb{C}}
\newcommand{\R}{\mathbb{R}}

\newcommand{\Ztre}{\Z/3\Z}
\newcommand{\wG}{\widehat{G}}
\newcommand{\GL}{\mathrm{GL}(2,\C)}
\newcommand{\SL}{\mathrm{SL}(2,\C)}
\newcommand{\SLL}{\mathrm{SL}}

\newcommand{\Ho}{\mathrm{H}}
\newcommand{\Aut}{\mathrm{Aut}}
\newcommand{\Lie}{\mathrm{Lie}}
\newcommand{\mf}{\mathfrak}
\newcommand{\mc}{\mathcal}
\newcommand{\g}{\mf{g}}
\newcommand{\h}{\mf{h}}
\renewcommand{\u}{\mf{u}}

\newcommand{\gl}{\mf{gl}}
\newcommand{\sll}{\mf{sl}}
\newcommand{\gC}{\mf{C}}

\renewcommand{\a}{\mf{a}}
\newcommand{\e}[1]{\text{\small$|#1\rangle$}}
\newcommand{\se}[1]{\text{\tiny$|#1\rangle$}}
\newcommand{\cV}{\mc{V}}
\newcommand{\cF}{\mc{F}}

\newcommand{\sG}{\mc{G}}

\newcommand{\mye}[1]{\text{\small$|#1\rangle$}}
\newcommand{\z}{\mf{z}}

\newcommand{\diag}{\mathrm{diag}}
\newcommand{\ad}{\mathrm{ad}}
\newcommand{\id}{\mathrm{id}}

\numberwithin{equation}{section}
\newtheorem{theorem}{Theorem}[section]

\newtheorem{crit}[theorem]{Criterion}
\newtheorem{lemma}[theorem]{Lemma}

\newtheorem{defn}[theorem]{Definition}            
\theoremstyle{remark}
 
\theoremstyle{remark}
\newtheorem{rmk}[theorem]{Remark}

\title{Classification of real and complex 3-qutrit states}

\author[S. Di Trani]{Sabino Di Trani}
\address[Di Trani]{Dipartimento di Matematica "Guido Castelnuovo", Sapienza - Università di Roma, Roma, Italy}
\author[W. A. de Graaf]{Willem A.\ de Graaf}
\address[de Graaf]{Department of Mathematics, University of Trento, Povo (Trento), Italy}
\author[A. Marrani]{Alessio Marrani}
\address[Marrani]{Instituto de F\'\i sica Teorica, Departamento de F{\'{i}}sica, Univ. de Murcia, Campus de Espinardo, Murcia, Spain}
\email{\rm sabino.ditrani@uniroma1.it, willem.degraaf@unitn.it,alessio.marrani@um.es}

\date{}

\begin{document}

\maketitle

\begin{abstract}
In this paper we classify the orbits of the group $\SLL(3,F)^3$ on the
space $F^3\otimes F^3\otimes F^3$ for $F=\C$ and $F=\R$. This is known as the
classification of complex and real 3-qutrit states. We also give an overview
of physical theories where these classifications are relevant.
\end{abstract}

\section{Introduction}\label{Sect:Int}

Elements of the vector space $(\C^3)^{\otimes n}=\C^3\otimes \cdots \otimes \C^3$ ($n$ factors
$\C^3$) are called {\em $n$-qutrit states}. On this space acts the
group $\SLL(3,\C)^n$, also called the SLOCC-group (Stochastic Local Quantum
Operations assisted by Classical Communication). One of the basic problems
in this context is to determine the orbits of this group, or in other words,
to determine the classification of the $n$-qutrit states under
SLOCC-equivalence. This has applications in quantum information theory,
as well as in high energy particle physics, see Section \ref{sec:cq}.

For $n=3$ the group and its representation can be constructed using a
grading of the simple Lie algebra of type $\mathfrak{e}_6$ (see Section
\ref{sec:E6} for the details). As a consequence,
Vinberg's theory of $\theta$-groups (\cite{vinberg,vinberg2}) can be used
to classify the orbits. This has been carried out by Nurmiev in \cite{Nurmiev}.
In this paper we revisit his classification, adding many details such
as descriptions of the stabilizers of the orbit representatives.

One of the main features of the construction of the representation using
a grading is that the elements of $(\C^3)^{\otimes 3}$ are shown to have a Jordan
decomposition. That is, each $v\in (\C^3)^{\otimes 3}$ can uniquely be
written as $v=s+e$, where $s,e\in (\C^3)^{\otimes 3}$, $s$ is {\em semisimple},
$e$ is {\em nilpotent} and $[s,e]=0$ (the latter is the Lie product in the
simple Lie algebra of type $\mathfrak{e}_6$). Here the terms semisimple and
nilpotent can be defined in the Lie algebra of type $\mathfrak{e}_6$. It is also
possible to define them more geometrically as follows. An element $s\in
(\C^3)^{\otimes 3}$ is semisimple if its orbit under $\SLL(3,\C)^3$ is closed
(in the Zarisky topology). An element $e\in (\C^3)^{\otimes 3}$ is nilpotent if
the closure of its orbit under $\SLL(3,\C)^3$ contains 0. An element is said
to be {\em mixed} if it is neither semisimple nor nilpotent. In this paper
we show the following result. 

\begin{theorem}
There are 90 classes of elements of $(\C^3)^{\otimes 3}$ with the following
properties. Each element of $(\C^3)^{\otimes 3}$ is $\SLL(3,\C)^3$-conjugate
to an element of precisely one class. All elements of a fixed class have
the same stabilizer in $\SLL(3,\C)^3$. We have 63 classes consisting of one
nilpotent element (see Table \ref{CNilp}). There are 4 classes consisting of
infinitely many semisimple elements and depending on respectively 3,2,1,1
parameters (see Section \ref{sec:semsimC}). 
Elements of one of these classes are $\SLL(3,\C)^3$-conjugate
if and only if they are conjugate under an explicitly given finite group,
of order 648, 18, 6, 3 respectively. Finally, there are in total 23 classes
of mixed elements (see Tables \ref{NilpotentMixed2}, \ref{NilpotentMixed3},
\ref{NilpotentMixed4}).
\end{theorem}  

We also consider the real field. The elements of
$(\R^3)^{\otimes n}=\R^3\otimes \cdots \otimes \R^3$ ($n$ factors
$\R^3$) are called {\em $n$-retrit states}. On this space we consider the
action of the group $\SLL(3,\R)^n$. This has applications in real quantum
mechanics, as well as to a ``black hole \& black string /qutrit correspondence''
in high-energy theoretical physics, see Section \ref{sec:rq}.

The main result of the present paper is the classification of the
$\SLL(3,\R)^3$-orbits on the space $(\R^3)^3$. To obtain this classification
we use the classification of the complex orbits as well as Galois
cohomological methods developed in \cite{BorovoideGraafLe,BorovoideGraafLe2}.
We summarize our findings for the real case in the following theorem.

\begin{theorem}
There are 109 classes of elements of $(\R^3)^{\otimes 3}$ with the following
properties. Each element of $(\R^3)^{\otimes 3}$ is $\SLL(3,\R)^3$-conjugate
to an element of precisely one class. The elements of a fixed class all have
the same stabilizer in $\SLL(3,\R)^3$. We have 70 classes consisting of one
nilpotent element (the elements of Table \ref{CNilp}, along with Table
\ref{RNilp}). There are 6 classes consisting of infinitely many semisimple
elements (these are the families of Section \ref{sec:semsimC}, as well as
the two families found in Section \ref{semsim:real}). 
Finally, there are in total 33 classes of mixed elements (the elements of
Tables \ref{NilpotentMixed2}, \ref{NilpotentMixed3}, \ref{NilpotentMixed4},
the elements given in \eqref{eq:mixedcan}, and those in Tables
\ref{tab:mixedR1}, \ref{tab:mixedR2}).
\end{theorem}
\subsection{Organization of the Paper:} Section~\ref{sec:q3ts} is devoted to giving a general physical motivation for studying complex and real qutrits, while in Section~\ref{sec:BH} we review the physical models involving $3$-qutrits. 

In the two subsequent sections we give a description of the mathematical techniques we used to obtain our computations. Using Vinberg theory of $\theta$-groups, in Section~\ref{sec:E6} we describe how  the $\SLL(3,\C)^3$-module $\C^3 \otimes \C^3 \otimes \C^3$ can be constructed starting from a suitable $\Z/3\Z$-grading on the Lie algebra $\mf{e}_6$, induced by an order $3$ automorphism.
%As a consequence, it results that $\SLL(3,\C)^3$-orbits in the module $\C^3 \otimes \C^3 \otimes \C^3$ can be divided in three families: nilpotent, semisimple and mixed.
In Section~\ref{sec:galois} we briefly recall how Galois cohomology can be used to achieve a description of real orbits, starting from a complete classification of the complex ones. 

Finally, in Sections~\ref{sec:nil}, \ref{semis} and \ref{sec:mixed} we give the details of our
computations that yield the classifications of the real and complex orbits, as well a the 
tables listing representatives of those orbits.
\section{Qutrits}\label{sec:q3ts}

\subsection{Complex qutrits}\label{sec:cq}

A quantum trit (\textit{qutrit}) is the superposition of three orthogonal
basis states $\left\{ \left\vert 0\right\rangle ,\left\vert 1\right\rangle
,\left\vert 2\right\rangle \right\} $ \cite{1-wiki}, rather than the two
which characterize a quantum bit (\textit{qubit}), namely :%
\begin{equation}
\left\vert \psi \right\rangle =\alpha \left\vert 0\right\rangle +\beta
\left\vert 1\right\rangle +\gamma \left\vert 2\right\rangle ,~\text{with~}%
\alpha ,\beta ,\gamma \in \mathbb{C}:\left\vert \alpha \right\vert
^{2}+\left\vert \beta \right\vert ^{2}+\left\vert \gamma \right\vert ^{2}=1.
\end{equation}%
Therefore, a qutrit is an element of $\mathbb{C}^{3}$ and the
corresponding SLOCC group acting on it is $GL_{3}(\mathbb{C})$.
The entanglement measures for qutrit systems are provided by relative
invariants under the action of the multipartite local SLOCC group. For
example, as the simplest relative SLOCC invariant for a two-qutrit state of
the form ($A,B=0,1,2$)%
\begin{equation}
\left\vert \Psi \right\rangle =\Psi _{AB}\left\vert A\right\rangle \otimes
\left\vert B\right\rangle ,
\end{equation}%
one can take the determinant det$(\Psi )$ of the $3\times 3$ complex matrix $%
\Psi _{AB}$, which is invariant under $\SLL(3,\mathbb{C})\otimes \SLL(3,
\mathbb{C})$. The classification of two-qutrit states is thus very simple :
different SLOCC classes are labelled by the rank of $\Psi_{AB}$. An
operator representation of the qutrit density matrix has been developed, and
qutrit entanglement has been studied; see e.g. in \cite{2}. Furthermore, the
generalized concurrence formula as a measure of two qutrits entanglement has
also been introduced \cite{9}.

Nowadays, there is ongoing work to develop quantum computers using not only
qubits, but also qutrits with multiple states \cite{2-wiki}. Indeed, qutrits
exhibit several peculiar and interesting features when used for storing
quantum information; for example, they are more robust to decoherence under
certain environmental interactions \cite{5-wiki}; however, in reality their
direct manipulation might be tricky, and it is thus often easier to use an
entanglement with a qubit \cite{6-wiki}. In \cite{Pal-Vertesi}, the maximum
quantum violation of over 100 tight bipartite Bell inequalities on systems
with up to four-dimensional Hilbert spaces was numerically determined, and
it was found that there exist Bell inequalities that can be violated more
with real qutrits than with complex qubits.

The physical implementation of a qutrit quantum computer in the context of
trapped ions has been studied \cite{7}, and quantum computer simulation
packages for qutrits have been implemented \cite{8}. New studies claim that
qutrits offer a promising path towards extending the frontier of quantum
computers \cite{36,37}. In fact, accessing the third state in current
quantum processors is primarily useful for researchers exploring the cutting
edge of quantum computing, quantum physics \cite{4-article} and those
interested in traditional, qubit-based algorithms alike. Promisingly,
qutrits cannot only increase the amount of information encoded in a single
element, but they can also enable techniques that can dramatically decrease
readout errors \cite{2-article}. Recent progress in world-leading,
cutting-edge quantum computer laboratories has shown how qutrit-qutrit gates
can reduce the cost of decomposing 3-qubit gates into basic 2-qubit
components \cite{3-article}. This is in part due to the much larger state
space accessible using qutrits : for instance, single qutrit operations
pertain to a 3-dimensional Hilbert space, while 2-qutrit
operations live in a 9-dimensional Hilbert space, which is more than twofold
larger than the four-dimensional Hilbert space of the 2-qubit case.

\subsubsection{Neutrinos as qutrits}

Recently, qutrits found an intriguing application in high energy particle
physics : indeed, neutrino flavour states, which are a superposition of
three states, have been characterized in terms of qutrits. This can be done
by mapping the density matrix for neutrinos to a \textit{generalized Poincar%
\'{e} sphere} \cite{2,24}.

The Poincar\'{e} sphere has its origin in optics and is a way of visualizing
different types of polarized light using the mapping from $SU_{2}$ to $S^{3}$%
. A qubit represents a point on the Poincar\'{e} sphere of $SU_{2}$ defined
as the complex projective line%
\begin{equation}
  P^{1}\C=\frac{SU_{2}}{U_{1}}.
\end{equation}%
A generalization of the Poincar\'{e} sphere to $SU_{3}$ can be constructed
\cite{27,28,29}, yielding to the characterization of qutrits as points on
the complex projective plane \cite{30, 3-wiki}%
\begin{equation}
P^{2}\C=\frac{SU_{3}}{U_{2}}.
\end{equation}%
An $n$-qutrit register can represent $3^{n}$ different states
simultaneously, i.e., a superposition state vector in $3^{n}$-dimensional
Hilbert space \cite{4-wiki}. Entanglement of neutrinos realized in terms of
Poincar\'{e} sphere representation for two resp. three-flavour neutrino
states using $SU_{2}$ Pauli matrices resp. $SU_{3}$ Gell-Mann matrices has
been investigated in \cite{Neutr-Qutrits}.

\subsection{Real qutrits: retrits}\label{sec:rq}

\textit{Real} quantum mechanics (that is, quantum mechanics defined over
real vector spaces) dates back to St\"{u}ckelberg \cite{Stueck}, and
provides an interesting theory, whose study may shed some light to
discriminate among the aspects of \textit{quantum entanglement} which are
unique to standard quantum theory and those which are more generic over
other physical theories endowed with this phenomenon \cite{real-QM}. Besides
\textit{rebits} (real qubits), in real quantum mechanics there has been an
interest in introducing \textit{retrits}, i.e. qutrits with \textit{real}
coefficients for probability amplitudes of a three-level system, namely a
three-level quantum state that may be expressed as a \textit{real} linear
combination of the orthonormal basis states $\left\vert 0\right\rangle $, $%
\left\vert 1\right\rangle $ and $\left\vert 2\right\rangle $; such a system
may be in a superposition of three mutually orthonormal quantum states
\cite{1-wiki}, which span the Hilbert space $\mathcal{H}^{3}$ \cite{3-wiki},
with dim$=3$ (over $\mathbb{C}$ for qutrits, or over $\mathbb{R}$ for
retrits). As for rebits, also for retrits the density matrix of the
processed quantum state $\rho $ is real; i.e., at each point in the quantum
computation, it holds that $\left\langle x|\rho |y\right\rangle \in \mathbb{R%
}$, \ for all $\left\vert x\right\rangle $, $\left\vert y\right\rangle $ in
the computational basis.

\subsubsection{Black hole/string charges as retrits}

In recent years, the relevance of rebits in high-energy theoretical physics
was highlighted by the determination of striking multiple relations between
the entanglement of pure states of 2 and 3 qubits and the entropy of extremal black holes
in $D=4$ Maxwell-Einstein-scalar theories, which can be regarded as bosonic
sectors of supergravity theories, or equivalently as low-energy limits of
string theory compactifications. In this framework, which has been
subsequently dubbed \textit{\textquotedblleft black hole / qubit
correspondence\textquotedblright } (see e.g. \cite{Oct, BH-1, BH-2} for
reviews and lists of Refs.), rebits acquire the physical meaning of the
electric and magnetic charges of the extremal black hole, and they linearly
transform under the generalized e.m. duality group $G(\mathbb{R})$ (named
U-duality group in string theory; see further below) of the theory under
consideration\footnote{%
In supergravity, the approximation of \textit{real} (rather than integer)
electric and magnetic charges of the black hole is often considered, thus
disregarding the charge quantization.}.

Thus, a natural question arises out : \textit{can one find entangled systems
of real qutrits that can be related to other entropy formulae?}

The answer is positive, provided one considers black hole and black string
solutions in $D=5$ space-time dimensions \cite{17-rev} : in this sense, the
aforementioned $D=4$ \textit{\textquotedblleft black hole / qubit
correspondence\textquotedblright } can be lifted to a\ $D=5$ \textit{%
\textquotedblleft black hole or black string / qutrit
correspondence\textquotedblright }.

It is here worth recalling that magic $\mathcal{N}=2$, $D=5$ supergravities
\cite{GST1, GST2, Oct} coupled to 5, 8, 14 and 26 vector multiplets, with $G(%
\mathbb{R})=SL(3,\mathbb{R})$, $SL(3,\mathbb{C})_{\mathbb{R}}$, $SU^{\ast
}(6)\simeq SL(3,\mathbb{H})_{\mathbb{R}}$ and $E_{6(-26)}\simeq SL(3,\mathbb{%
O})_{\mathbb{R}}$ \cite{Dray-or-Wangberg} respectively, can be described by
simple, rank-3 Jordan algebras $\mathfrak{J}_{3}^{\mathbb{A}}$ of $3\times 3$
Hermitian matrices with entries taken from the four normed Hurwitz algebras
\cite{Hurwitz}: the reals $\mathbb{R}$, complexes $\mathbb{C}$, quaternions $%
\mathbb{H}$ and octonions $\mathbb{O}$, of total dimension dim$_{\mathbb{R}%
}=3q+3$, with $q:=$dim$_{\mathbb{R}}\mathbb{A}=1,2,4,8$ for $\mathbb{A}=%
\mathbb{R}$, $\mathbb{C},\mathbb{H}$~and $\mathbb{O}$, respectively.

Moreover, one can also replace in these Jordan algebras the division
algebras by their split versions (i.e., split complexes $\mathbb{C}_{s}$,
split quaternions $\mathbb{H}_{s}$ and split octonions $\mathbb{O}_{s}$),
obtaining $\mathfrak{J}_{3}^{\mathbb{A}_{s}}$. For $\mathbb{C}_{s}$ and $%
\mathbb{H}_{s}$, respectively having $G(\mathbb{R})=SL(3,\mathbb{R})\otimes
SL(3,\mathbb{R})$ and $G(\mathbb{R})=SL(6,\mathbb{R})$, one obtains
non-supersymmetric (i.e., $\mathcal{N}=0$) Maxwell-Einstein-scalar theories
in $D=5$, which have been investigated in \cite{MRR1, MRR2}, whereas for $%
\mathbb{O}_{s}$ one obtains the maximally supersymmetric, $\mathcal{N}=8$, $%
D=5$ supergravity \cite{86-rev} with 27 Abelian gauge fields transforming in
the reflexive, real, defining module $\mathbf{27}$ of $G(\mathbb{R}%
)=E_{6(6)}$.

All these theories admit asymptotically flat, static, extremal black holes
and black strings, whose Bekenstein-Hawking entropy can respectively be
expressed in terms of the (absolute value of the) homogeneous cubic
polynomials $\mathbf{I}_{3,el}$ resp. $\mathbf{I}_{3,magn}$, which are
invariant under the non-transitive action of $G(\mathbb{R})$ over the
corresponding simple Jordan algebra $\mathfrak{J}_{3}$ \cite{FGimK}; indeed,
in all mentioned cases, $G$ acting on $\mathfrak{J}_{3}$ is a $\theta $%
-group in the sense of Vinberg \label{Vinberg-Weyl} (see e.g. also \cite%
{Kac-80}, and Refs. therein). On the quantum information theory (QIT) side,
such cubic polynomials describe the measure of a multipartite entanglement
of a suitable system of qutrits; for a comprehensive account, see e.g. Sec.
4 of the review \cite{BH-2}.

\section{From 3 qutrits to 3-centered black holes/strings}\label{sec:BH}

\subsection{\label{Sec-C}Complex orbits: entanglement of 3 qutrits}

In \cite{Nur00} the following $\mathbb{Z}/3\mathbb{Z}$-grading of $\mathfrak{%
e}_{6}(\mathbb{C})$ was considered:
\begin{eqnarray}
\mathfrak{e}_{6}(\mathbb{C}) &=&\mathfrak{g}_{-1}(\mathbb{C})\oplus
\mathfrak{g}_{0}(\mathbb{C})\oplus \mathfrak{g}_{1}(\mathbb{C});  \notag \\
\mathfrak{g}_{0}(\mathbb{C}) &=&\mathfrak{sl}_{3}(\mathbb{C})^{\oplus 3};
\notag \\
\mathfrak{g}_{1}(\mathbb{C}) &=&\mathbb{C}^{3}\otimes \mathbb{C}^{3}\otimes
\mathbb{C}^{3}\equiv \left( \mathbb{C}^{3}\right) ^{\otimes 3};  \notag \\
\mathfrak{g}_{-1}(\mathbb{C}) &=&\mathfrak{g}_{1}^{\ast }(\mathbb{C}),
\label{dec}
\end{eqnarray}%
and the orbit action of the group $G_{0}(\mathbb{C})=\SLL(3,\mathbb{C}%
)^{\times 3}$ acting on $\left( \mathbb{C}^{3}\right) ^{\otimes 3}$ was
completely classified. Note that $\mathfrak{g}_{0}(\mathbb{C})=\mathfrak{sl}%
_{3}(\mathbb{C})^{\oplus 3}$ is a maximal non-symmetric subalgebra of $%
\mathfrak{e}_{6}(\mathbb{C})$, defined as the fixed point algebra with
respect to the action of a periodic authomorphism of order 3, inducing a $%
\mathbb{Z}/3\mathbb{Z}$-grading on $\mathfrak{e}_{6}(\mathbb{C})$ itself.
Moreover, $\SLL(3,\mathbb{C})^{\times 3}$ (non-transitively) acting on $%
\left( \mathbb{C}^{3}\right) ^{\otimes 3}$ is a $\theta $-group of type II
in the sense of Vinberg \cite{Vinberg-Weyl}; indeed, it has a finite number
of nilpotent orbits, and its ring of invariant polynomials has dimension 3,
being finitely generated by the homogeneous polynomials of the integrity
basis, of degree 6, 9 and 12 (resp. named $\mathbf{I}_{6}$, $\mathbf{I}_{9}$
and $\mathbf{I}_{12}$ below; see e.g. Table III of \cite{Kac-80}, and Refs.
therein).

The above grading has application in Quantum Information Theory (QIT),
because it concerns the classification of the entanglement of 3 pure
multipartite qutrits, fitting the space $\left(
\mathbb{C}^{3}\right) ^{\otimes 3}\equiv \left( \mathbf{3},\mathbf{3},%
\mathbf{3}\right) $ (in physicists' notation), under the SLOCC group $\SLL(3,
\mathbb{C})^{\times 3}$. For a recent account on the
classification of the entanglement classes of 3 qutrits, see e.g. \cite%
{Gharahi-Mancini}, and Refs. therein.

\subsection{Real orbits: 3-centered black holes/strings in $D=5$}
\label{sec:bhD5}

The maximally non-compact (i.e., split) real form of the $\mathbb{Z}/3%
\mathbb{Z}$-grading (\ref{dec}) reads%
\begin{eqnarray}
\mathfrak{e}_{6(6)} &=&\mathfrak{g}_{-1}(\mathbb{R})\oplus \mathfrak{g}_{0}(%
\mathbb{R})\oplus \mathfrak{g}_{1}(\mathbb{R});  \notag \\
\mathfrak{g}_{0}(\mathbb{R}) &=&\mathfrak{sl}_{3}(\mathbb{R})^{\oplus 3};
\notag \\
\mathfrak{g}_{1}(\mathbb{R}) &=&\left( \mathbb{R}^{3}\right) ^{\otimes 3};
\notag \\
\mathfrak{g}_{-1}(\mathbb{R}) &=&\mathfrak{g}_{1}^{\prime }(\mathbb{R}).
\label{dec-split}
\end{eqnarray}%
According to the classification carried out in \cite{dGM1}, the split form $%
\mathfrak{e}_{6(6)}$ admits only two real forms of $\mathfrak{a}_{2}^{\oplus
3}$ as maximal (non-symmetric) subalgebras :%
\begin{equation}
\mathfrak{e}_{6(6)}\supsetneq \left\{
\begin{array}{ll}
i: & \mathfrak{sl}_{3}(\mathbb{R})^{\oplus 3}; \\
ii: & \mathfrak{su}_{1,2}\oplus \mathfrak{sl}_{3}(\mathbb{C})_{\mathbb{R}}.%
\end{array}%
\right.  \label{e6(6)}
\end{equation}%
On the other hand, analougous embeddings exist for all other non-compact,
real forms of $\mathfrak{e}_{6}(\mathbb{C})$, namely \cite{dGM1} :%
\begin{eqnarray}
\mathfrak{e}_{6(2)} &\supsetneq &\left\{
\begin{array}{ll}
i: & \mathfrak{su}_{1,2}^{\oplus 3}; \\
ii: & \mathfrak{su}_{1,2}\oplus \mathfrak{su}_{3}^{\oplus 2}; \\
iii: & \mathfrak{sl}_{3}(\mathbb{R})\oplus \mathfrak{sl}_{3}(\mathbb{C})_{%
\mathbb{R}};%
\end{array}%
\right.  \label{e6(2)} \\
\mathfrak{e}_{6(-14)} &\supsetneq &\mathfrak{su}_{1,2}^{\oplus 2}\oplus
\mathfrak{su}_{3};  \label{e6(-14)} \\
\mathfrak{e}_{6(-26)} &\supsetneq &\mathfrak{su}_{3}\oplus \mathfrak{sl}_{3}(%
\mathbb{C})_{\mathbb{R}}.  \label{e6(-26)}
\end{eqnarray}%
As embedding (\ref{e6(6)}) $i$ is related to the $\mathbb{Z}/3\mathbb{Z}$%
-grading (\ref{dec-split}) of $\mathfrak{e}_{6(6)}$, the question presents itself
as to whether the other real
embeddings above are related to suitable $\mathbb{Z}/3\mathbb{Z}$-gradings
of the corresponding non-compact, real forms of $\mathfrak{e}_{6}$. Here we
will not go into this question.

Some observations are in order.

\begin{itemize}
\item In (not necessarily supersymmetric) Maxwell-Einstein-scalar theories,
the embeddings (\ref{e6(6)}) $i$ and (\ref{e6(2)}) $iii$ are named \textit{%
super-Ehlers} \cite{SE} (or \textit{Jordan pairs} \cite{JP}) embeddings, and
they express the fact that the Lie algebra of the $D=3$ electric-magnetic
(e.m.) duality group\footnote{%
This group can be regarded as the \textquotedblleft continuous
limit/version\textquotedblright\ $G\left( \mathbb{R}\right) $ of
the U-duality group $G(\mathbb{Z})$ of superstring theory\cite{HT}.} contains
the Lie algebra of the $D\geqslant 4$ e.m. duality group as a proper
subalgebra, along with a commuting summand $\mathfrak{sl}_{D-2}(\mathbb{R})$,
specified for $D=5$.

\item The embeddings (\ref{e6(6)}) $i$ and $ii$ follow from consistent
truncations of the corresponding Maxwell-Einstein-scalar theories in $D=5$ :
the Lie algebra of the e.m. duality group of the ($\mathfrak{J}_{3}^{\mathbb{%
O}_{s}}$-based) \textit{maximal} supergravity contains (as proper
subalgebras) the Lie algebra of the e.m. duality group of $\mathcal{N}=0$ $%
\mathfrak{J}_{3}^{\mathbb{C}_{s}}$-based Maxwell-Einstein-scalar theory, as
well as of the $\mathcal{N}=2$ $\mathfrak{J}_{3}^{\mathbb{C}}$-based magic
supergravity, along with a commuting summand $\mathfrak{sl}_{3}(\mathbb{R})$
and $\mathfrak{su}_{1,2}$, respectively. These truncations are consequences
of division algebraic reductions (which are also supersymmetry reductions) :
$\mathbb{O}_{s}\rightarrow \mathbb{C}_{s}$ ($\mathcal{N}=8\rightarrow 0$)
and $\mathbb{O}_{s}\rightarrow \mathbb{C}$ ($\mathcal{N}=8\rightarrow 2$),
respectively.

\item On the other hand, the embedding (\ref{e6(-26)}) follows from a
consistent truncation of the corresponding supergravity theories : in $D=5$,
the Lie algebra of the e.m. duality group of the \textit{exceptional} $%
\mathcal{N}=2$, $J_{3}^{\mathbb{O}}$-based magic supergravity \cite{GST1,
GST2} contains (as proper subalgebra) the Lie algebra of the e.m. duality
group of the $\mathcal{N}=2$ $J_{3}^{\mathbb{C}}$-based magic supergravity,
along with a commuting summand $\mathfrak{su}_{3}$. This truncation is a
consequence of the division algebraic reduction $\mathbb{O}\rightarrow
\mathbb{C}$ (which implies no reduction of supersymmetry).\medskip
\end{itemize}

In the present paper, we will study the orbit stratification of the
non-transitive action on $\left( \mathbb{R}^{3}\right) ^{\otimes 3}$ of the
totally split real form $3\mathfrak{sl}_3(\R)$.
This corresponds to the embedding (\ref{e6(6)}) $i$, in which one $\mathfrak{sl}_{3}(\mathbb{R})\equiv \mathfrak{sl}_{3,h}(%
\mathbb{R})$ can be interpreted as the \textquotedblleft
horizontal\textquotedblright\ symmetry \cite{FMOSY} of (asymptotically flat)
3-centered extremal black hole/string solutions in the corresponding
(ungauged) Maxwell-Einstein-scalar (super)gravity theory in $D=5$.

In fact, \textit{multi-centered} solutions (which are a natural
generalization of single-centered solutions) to Maxwell-Einstein equations
exist also in $D>4$ (Lorentzian) space-time dimensions \cite{Lucietti-multi}%
. By considering $D=5$, the asymptotically flat black objects are 0-branes
(black \textit{holes}) and 1-branes (black \textit{strings}); in their
extremal cases, their near-horizon geometries respectively are $%
AdS_{2}\otimes S^{3}$ and $AdS_{3}\otimes S^{2}$. Asymptotically flat,
extremal black holes in $D=5$ are named Tangherlini black holes \cite{T1,T2}%
, and they are a special instance of the Papapetrou-Majumdar black holes,
whose multi-centered classes have been recently investigated in \cite%
{Lucietti-multi}, in which the classification of asymptotically flat,
extremal black hole solutions in Maxwell-Einstein theory in higher
dimensions has been completed.

In $D=5$ the fluxes of the two-form Abelian (electric) field-strengths and
of their dual three-form Abelian (magnetic) field-strengths, which
respectively determine the \textit{electric} charges of the extremal
Tangherlini black hole and the \textit{magnetic} charges of the extremal
black string, fit into a representation $\mathcal{R}_{el}$ resp. into its
dual/conjugate representation $\mathcal{R}_{magn}\equiv \mathcal{R}%
_{el}^{\prime }$ of the $D=5$ e.m. duality group $G\left( \mathbb{R}\right) $. In
the $D=5$ Maxwell-Einstein-scalar theories pertaining to the real embeddings
(\ref{e6(6)}) $i$ resp. (\ref{e6(2)}) $iii$ relevant to our treatment, the
following group-theoretical data can be specified :

\begin{itemize}
\item $\mathcal{N}=0$ Maxwell-Einstein-scalar theory based on $\mathfrak{J}%
_{3}^{\mathbb{C}_{s}}$ : $G\left( \mathbb{R}\right) =\SLL(3,\mathbb{R}%
)\times \SLL(3,\mathbb{R})$, $\mathcal{R}_{el}=\left( \mathbf{3},\mathbf{3}%
^{\prime }\right) $, $\mathcal{R}_{magn}=\left( \mathbf{3}^{\prime },\mathbf{%
3}\right) $, and each orbit of $\SLL(3,\mathbb{R})^{\times 2}$ acting on $%
\mathcal{R}_{el}$ resp. $\mathcal{R}_{magn}$ supports a unique class of
single-centered (electric black hole resp. magnetic black string) solutions
\cite{MRR1,MRR2}.

\item $\mathcal{N}=2$ Maxwell-Einstein magic supergravity based on $%
\mathfrak{J}_{3}^{\mathbb{C}}$ : $G\left( \mathbb{R}\right) =\SLL(3,\mathbb{C%
})_{\mathbb{R}}$, $\mathcal{R}_{el}=\left( \mathbf{3},\overline{\mathbf{3}}%
\right) $, $\mathcal{R}_{magn}=\left( \overline{\mathbf{3}},\mathbf{3}%
\right) $, and each orbit of $\SLL(3,\mathbb{C})_{\mathbb{R}}$ acting on $%
\mathcal{R}_{el}$ resp. $\mathcal{R}_{magn}$ supports a unique class of
single-centered (electric black hole resp. magnetic black string) solutions
\cite{small-orbits}.
\end{itemize}

From classical invariant theory, for both the aforementioned gravity
theories in $D=5$, the action of the e.m. duality group $G\left( \mathbb{%
R}\right) $ on $\mathcal{R}_{el}$ or $\mathcal{R}_{magn}$ can be traced
back, on $\mathbb{C}$, to the action of the $\theta $-group of type I $%
\SLL(3,\mathbb{C})\times \SLL(3,\mathbb{C})$ on its bi-fundamental
representation $\mathbb{C}^{3}\otimes \mathbb{C}^{3}$, characterized by the
following properties (see e.g. \cite{Kac-80}, and Refs. therein):

\begin{description}
\item[1] the identity-connected component $\left( H_{m=1}\right) _{0}$ of
the stabilizer of the generic orbit $\mathcal{O}_{m=1}=\SLL(3,\mathbb{C}%
)^{\times 2}/H_{m=1}$ is given by $\SLL(3,\mathbb{C})$;

\item[2] the number of nilpotent orbits is finite;

\item[3] the ring of invariant polynomials is 1-dimensional, and it is
finitely generated by a cubic homogeneous polynomial $\mathbf{I}_{3}$ (see
e.g. Table II of \cite{Kac-80}). Such a cubic polynomial is nothing but the
cubic norm (determinant) of the rank-3 simple Jordan algebra $\left( J_{3}^{%
\mathbb{C}}\right) _{\mathbb{C}}\simeq \mathbb{C}^{3}\otimes \mathbb{C}^{3}$
of $\SLL(3,\mathbb{C})\times \SLL(3,\mathbb{C})$, which can be regarded as
the \textit{reduced structure group} of $\left( J_{3}^{\mathbb{C}}\right) _{%
\mathbb{C}}$ itself. The same will hold for the relevant forms on $\mathbb{R}$,
\textit{mutatis mutandis}.

\item[3.1] In the $\mathcal{N}=0$ Maxwell-Einstein-scalar theory based on $%
\mathfrak{J}_{3}^{\mathbb{C}_{s}}$ \cite{MRR1,MRR2}: $G\left( \mathbb{R}%
\right) =\SLL(3,\mathbb{R})\times \SLL(3,\mathbb{R})$ acting on $\mathcal{R}%
_{el}=\left( \mathbf{3},\mathbf{3}^{\prime }\right) $ resp. $\mathcal{R}%
_{magn}=\left( \mathbf{3}^{\prime },\mathbf{3}\right) $ admits a unique
primitive invariant cubic homogeneous polynomial : $\mathbf{I}_{3,el}$ resp.
$\mathbf{I}_{3,magn}$, which is nothing but the cubic norm (determinant) of
the rank-3 simple Jordan algebra $J_{3}^{\mathbb{C}_{s}}\simeq \mathbb{R}%
^{3}\otimes \mathbb{R}^{3}$ (resp. its dual) of $\SLL(3,\mathbb{R})\times
\SLL(3,\mathbb{R})$, which can be regarded as the \textit{reduced structure
group} of $J_{3}^{\mathbb{C}_{s}}$ itself. Physically, $\mathbf{I}_{3,el}$
resp. $\mathbf{I}_{3,magn}$ determine the Bekenstein-Hawking entropy of the
extremal, 1-centered, static, asymptotically flat black hole resp. black
string solutions of the theory \cite{MRR1,MRR2}.

\item[3.2] In the $\mathcal{N}=2$ magic Maxwell-Einstein supergravity based
on $\mathfrak{J}_{3}^{\mathbb{C}}$ \cite{GST1, GST2}: $G\left( \mathbb{R}%
\right) =\SLL(3,\mathbb{C})_{\mathbb{R}}$ acting on $\mathcal{R}_{el}=\left(
\mathbf{3},\overline{\mathbf{3}}\right) $ resp. $\mathcal{R}_{magn}=\left(
\overline{\mathbf{3}},\mathbf{3}\right) $ admits a unique primitive
invariant cubic homogeneous polynomial : $\mathbf{I}_{3,el}$ resp. $\mathbf{I%
}_{3,magn}$, which is nothing but the cubic norm (determinant) of the rank-3
simple Jordan algebra $J_{3}^{\mathbb{C}}\simeq \mathbb{R}^{3}\otimes
\mathbb{R}^{3}$ (resp. its dual) of $\SLL(3,\mathbb{C})_{\mathbb{R}}$, which
can be regarded as the \textit{reduced structure group} of $J_{3}^{\mathbb{C}%
}$ itself. Physically, $\mathbf{I}_{3,el}$ resp. $\mathbf{I}_{3,magn}$
determine the Bekenstein-Hawking entropy of the extremal, 1-centered,
static, asymptotically flat black hole resp. black string solutions of the
theory \cite{FG1,FG2}.\medskip
\end{description}

As in $D=4$, a crucial feature of multi-centered solutions is the existence
of a global, \textit{\textquotedblleft horizontal\textquotedblright }
symmetry group\footnote{%
Actually, the \textquotedblleft horizontal\textquotedblright\ symmetry group
is $GL_{m}(\mathbb{R})$, where the additional scale symmetry $SO_{1,1}$ with
respect to $\SLL(m,\mathbb{R})$ is encoded by the homogeneity of the $G(%
\mathbb{R})$-invariant polynomials in the black hole electric (or black
string magnetic) charges. The subscript \textquotedblleft $h$%
\textquotedblright\ stands for \textquotedblleft
horizontal\textquotedblright\ throughout.} $\SLL_{m,h}(\mathbb{R})$ \cite%
{FMOSY}, encoding the combinatoric structure of the $m$-centered solutions
of the theory ($m\in \mathbb{N}$), and commuting with $G\left( \mathbb{R}%
\right) $ itself. Thus, in presence of an asymptotically flat $D=5$
multi-centered\footnote{%
We are not considering the intriguing possibility in which a $m$-centered
solution in $D=5$ is composed by $m^{\prime }$ black holes and $m-m^{\prime }
$ black strings (or the other way around), since we are currently not
knowledgeable whether this can be a well-defined solution of the equations
of motion of the theories under consideration.} black $p$-brane (with $p=0$
resp. $1$) solution with $m$ centers, the dimension $I_{m}$ of the ring of $%
\left( SL_{m,h}(\mathbb{R})\times G\left( \mathbb{R}\right) \right) $%
-invariant homogeneous polynomials constructed with $m$ distinct copies of
the $G\left( \mathbb{R}\right) $-representation $\mathcal{R}_{el}$ resp. $%
\mathcal{R}_{magn}$ (in the defining irrepr. $\mathbf{m}$ of $SL_{m,h}(%
\mathbb{R})$) is given by the general formula \cite{FMOSY}%
\begin{equation}
I_{m}=m\text{dim}_{\mathbb{R}}\mathcal{R}-\text{dim}_{\mathbb{R}}\mathcal{O}%
_{m},  \label{counting}
\end{equation}%
where%
\begin{equation}
\mathcal{O}_{m}:=\frac{SL_{m,h}(\mathbb{R})\times G\left( \mathbb{R}\right)
}{H_{m}\left( \mathbb{R}\right) }
\end{equation}%
is a generally non-symmetric coset describing the generic, open $\left(
SL_{m,h}(\mathbb{R})\times G\left( \mathbb{R}\right) \right) $-orbit,
spanned by the $m$ copies of the representation $\mathcal{R}_{el}$ resp. $%
\mathcal{R}_{magn}$, each pertaining to one center of the multi-centered
solution. By setting $m=1$ resp. $m=3$ in the aforementioned $D=5$
Maxwell-Einstein-scalar theories, one obtains :%
\begin{eqnarray}
\left.
\begin{array}{rr}
\underset{G\left( \mathbb{R}\right) =\SLL(3,\mathbb{R})^{\times 2}}{\mathcal{%
N}=0} & \mathfrak{J}_{3}^{\mathbb{C}_{s}} \\
\underset{G\left( \mathbb{R}\right) =\SLL(3,\mathbb{C})_{\mathbb{R}}}{%
\mathcal{N}=2} & \mathfrak{J}_{3}^{\mathbb{C}}%
\end{array}%
\right\}  &:&%
\begin{array}{l}
I_{m=1}=\text{dim}_{\mathbb{R}}\left( \left( \mathbb{R}^{3}\right) ^{\otimes
2}\right) \mathbf{-}\text{dim}_{\mathbb{R}}\left( \mathcal{O}_{m=1}\right)
=3\cdot 3-8=\underset{\mathbf{I}_{3}}{1}; \\
\\
I_{m=3}=3\text{dim}_{\mathbb{R}}\left( \left( \mathbb{R}^{3}\right)
^{\otimes 2}\right) \mathbf{-}\text{dim}_{\mathbb{R}}\left( \mathcal{O}%
_{m=3}\right) =3\cdot 9-3\cdot 8=\underset{\mathbf{I}_{6},\mathbf{I}_{9},%
\mathbf{I}_{12}}{3},%
\end{array}
\notag \\
&&  \label{ccounting}
\end{eqnarray}%
because%
\begin{eqnarray*}
\text{dim}_{\mathbb{R}}\mathcal{R}_{el} &=&\text{dim}_{\mathbb{R}}\mathcal{R}%
_{magn}=9; \\
\left( H_{m=1}\left( \mathbb{R}\right) \right) _{0} &=&\text{Lie
group~whose~Lie~algebra is a real form of }\mathfrak{a}_{2}\Rightarrow \text{%
dim}_{\mathbb{R}}\mathcal{O}_{m=1}=8; \\
\left( H_{m=3}\left( \mathbb{R}\right) \right) _{0} &=&\mathbb{I~}\text{%
(identity)}\Rightarrow \mathcal{O}_{m=3}\simeq SL_{3,h}(\mathbb{R})\times
G\left( \mathbb{R}\right) \Rightarrow \text{dim}_{\mathbb{R}}\mathcal{O}%
_{m=3}=24.
\end{eqnarray*}%
The counting (\ref{ccounting}) for $m=3$ implies that the ring of $\left(
SL_{3,h}(\mathbb{R})\times G\left( \mathbb{R}\right) \right) $-invariant
homogeneous polynomials built out of three copies of $\mathcal{R}_{el}$($%
=\left( \mathbf{3},\mathbf{3}^{\prime }\right) $ resp. $\left( \mathbf{3},%
\overline{\mathbf{3}}\right) $ for $G\left( \mathbb{R}\right) =\SLL(3,%
\mathbb{R})^{\times 2}$ resp. $\SLL(3,\mathbb{C})_{\mathbb{R}}$) or of $%
\mathcal{R}_{magn}$($=\left( \mathbf{3}^{\prime },\mathbf{3}\right) $ resp. $%
\left( \overline{\mathbf{3}},\mathbf{3}\right) $ for $G\left( \mathbb{R}%
\right) =\SLL(3,\mathbb{R})^{\times 2}$ resp. $\SLL(3,\mathbb{C})_{\mathbb{R}%
}$) has (real) dimension $3$. \ By taking the complexification of all this,
one retrieves the known result from classical invariant theory, namely that $%
\SLL(3,\mathbb{C})^{\times 3}$ non-transitively acting on the
tri-fundamental representation space $\mathbb{C}^{3}\otimes \mathbb{C}%
^{3}\otimes \mathbb{C}^{3}$ is a $\theta $-group of type II in the sense of
Vinberg \cite{Vinberg-Weyl}, and, as anticipated in\ Sec. \ref{Sec-C}, it
has the following properties :

\begin{enumerate}
\item the identity-connected component $\left( H_{m=3}\right) _{0}$ of the
stabilizer of the generic 3-centered orbit%
\begin{equation}
\mathcal{O}_{m=3}:=\frac{SL_{3,h}(\mathbb{R})\times G\left( \mathbb{R}%
\right) }{H_{m=3}\left( \mathbb{R}\right) }
\end{equation}
is nothing but $\mathbb{I}$ (identity);

\item the number of nilpotent orbits is finite;

\item the ring of invariant polynomials is 3-dimensional, and it is finitely
generated by an integrity basis of homogeneous polynomials $\mathbf{I}_{6}$,
$\mathbf{I}_{9}$ and $\mathbf{I}_{12}$ resp. of degree 6, 9, 12 (see e.g.
Table III of \cite{Kac-80}, and Refs. therein).\medskip
\end{enumerate}

Some further observations are in order.

\begin{itemize}
\item For simplicity's sake, let us assume to work on $\mathbb{C}$. Since
the \textquotedblleft horizontal\textquotedblright\ Lie algebra $\mathfrak{a}%
_{2,h}$ stands on a different footing with respect to the Lie algebra $%
\mathfrak{a}_{2}\oplus \mathfrak{a}_{2}$ of the e.m. duality group $G(%
\mathbb{C})=SL(3,\mathbb{C})^{\times 2}$, only invariance with respect to
the proper subgroup Sym$_{2}$ of the permutation group Sym$_{3}$
of the 3 tensor factors in $\mathbb{C}^{3}\otimes \mathbb{C}^{3}\otimes
\mathbb{C}^{3}$ should be taken into account, when considering 3-centered
black hole/string solutions in the $D=5$ Maxwell-Einstein (super)gravity
theories under consideration. In such a way, a classification invariant under%
\begin{equation}
\text{Sym}_{2}\ltimes \left( SL_{3,h}(\mathbb{C})\times \left( \SLL(3,%
\mathbb{C})^{\times 2}\right) \right)
\end{equation}
(namely, Sym$_{2}\ltimes \left( \mathfrak{a}_{2,h}\oplus \mathfrak{a}%
_{2}\oplus \mathfrak{a}_{2}\right) $ at Lie algebra level) should be
considered. Therefore, besides the choice of the ground number field ($%
\mathbb{C}$ in QIT \textit{versus} $\mathbb{R}$ in gravity), we are pointing
out another important difference in the orbit classification procedure :
while in QIT one needs to classify the action of
\begin{equation}
\text{Sym}_{3}\ltimes SL(3,\mathbb{C})^{\times 3}~~\text{on~~}\mathbb{C}%
^{3}\otimes \mathbb{C}^{3}\otimes \mathbb{C}^{3}\equiv \left( \mathbf{3},%
\mathbf{3},\mathbf{3}\right) ,
\end{equation}
in $D=5$ Maxwell-Einstein (super)gravity one needs to classify the action of%
\begin{gather}
\text{Sym}_{2}\ltimes \left( SL_{3,h}(\mathbb{R})\times \left( \SLL(3,%
\mathbb{R})^{\times 2}\right) \right) ~~\text{on~~}\mathbb{R}^{3}\otimes
\mathbb{R}^{3}\otimes \mathbb{R}^{3}\equiv \left( \mathbf{3},\mathbf{3},%
\mathbf{3}^{\prime }\right) \\
\text{resp.}  \notag \\
\text{Sym}_{2}\ltimes \left( SL_{3,h}(\mathbb{R})\times \SLL(3,\mathbb{C})_{%
\mathbb{R}}\right) ~~\text{on~~}\mathbb{R}^{3}\otimes \mathbb{R}^{3}\otimes
\mathbb{R}^{3}\equiv \left( \mathbf{3},\mathbf{3},\overline{\mathbf{3}}%
\right)
\end{gather}
in the $\mathcal{N}=0$, $\mathfrak{J}_{3}^{\mathbb{C}_{s}}$-based theory
\cite{MRR1, MRR2} resp. the $\mathcal{N}=2$, $\mathfrak{J}_{3}^{\mathbb{C}}$%
-based magic supergravity \cite{GST1, GST2} in $D=5$.

\item In the $\mathcal{N}=0$, $\mathfrak{J}_{3}^{\mathbb{C}_{s}}$-based
Maxwell-Einstein theory resp. the $\mathcal{N}=2$, $\mathfrak{J}_{3}^{%
\mathbb{C}}$-based magic supergravity, it holds that $\mathcal{R}%
_{el}=\left( \mathbf{3},\mathbf{3}^{\prime }\right) $ resp. $\left( \mathbf{3%
},\overline{\mathbf{3}}\right) $, and $\mathcal{R}_{magn}=\left( \mathbf{3}%
^{\prime },\mathbf{3}\right) $ resp. $\left( \overline{\mathbf{3}},\mathbf{3}%
\right) $. Therefore, as given above, \textit{modulo the action of
permutation groups}, the whole representation space on which $%
SL_{3,h}(\mathbb{R})\times G(\mathbb{R})$ (non-transitively) acts is $%
\mathbb{R}^{3}\otimes \mathcal{R}_{el}\equiv \left( \mathbf{3},\mathcal{R}%
_{el}\right) $ or $\mathbb{R}^{3}\otimes \mathcal{R}_{magn}\equiv \left(
\mathbf{3},\mathcal{R}_{magn}\right) $. The complexification of all this
yields an interpretation of the representation space $\mathbb{C}^{3}\otimes
\mathbb{C}^{3}\otimes \mathbb{C}^{3}$ as the $\left( \mathbf{3},\mathbf{3},%
\overline{\mathbf{3}}\right) $ or the $\left( \mathbf{3},\overline{\mathbf{3}%
},\mathbf{3}\right) $ of $\SLL(3,\mathbb{C})^{\times 3}$, which is different
from the interpretation of $\mathbb{C}^{3}\otimes \mathbb{C}^{3}\otimes
\mathbb{C}^{3}$ holding in QIT, namely as the $\left( \mathbf{3},\mathbf{3},%
\mathbf{3}\right) $ of the SLOCC group $\SLL(3,\mathbb{C})^{\times 3}$.
However, one can realize that the classification of the orbit action of $%
\SLL(3,\mathbb{C})^{\times 3}$ over $\left( \mathbf{3},\mathbf{3},\mathbf{3}%
\right) $ is the same as the one over $\left( \mathbf{3},\mathbf{3},%
\overline{\mathbf{3}}\right) $, see Section \ref{sec:333bar}.
By switching from $\SLL(3,\mathbb{C})^{\times 3}$ to the
split real form $\SLL(3,\mathbb{R})^{\times 3}$, this directly implies that
the classification of the orbit action of $\SLL(3,\mathbb{R})^{\times 3}$
over $\left( \mathbf{3},\mathbf{3},\mathbf{3}\right) $ is the same as the
one over $\left( \mathbf{3},\mathbf{3},\mathbf{3}^{\prime }\right) $.

\item One may also simply consider the $G\left( \mathbb{R}\right) $%
-invariant (electric resp. magnetic) polynomials constructed with $m$ copies
of $\mathcal{R}_{el}$ resp. $\mathcal{R}_{magn}$, and conceive the
\textquotedblleft horizontal\textquotedblright\ symmetry $SL_{m,h}(\mathbb{R}%
)$ merely as a multiplet-organizing (i.e., spectrum-generating) symmetry.
Setting $m=3$, the counting, rather than the second of (\ref{ccounting}),
now goes this way :%
\begin{equation}
\left.
\begin{array}{rr}
\underset{G\left( \mathbb{R}\right) =\SLL(3,\mathbb{R})^{\times 2}}{\mathcal{%
N}=0} & \mathfrak{J}_{3}^{\mathbb{C}_{s}} \\
\underset{G\left( \mathbb{R}\right) =\SLL(3,\mathbb{C})_{\mathbb{R}}}{%
\mathcal{N}=2} & \mathfrak{J}_{3}^{\mathbb{C}}%
\end{array}%
\right\} \text{:~}\tilde{I}_{m=3}=3\text{dim}_{\mathbb{R}}\left( \left(
\mathbb{R}^{3}\right) ^{\otimes 2}\right) \mathbf{-}\text{dim}_{\mathbb{R}%
}\left( G\left( \mathbb{R}\right) \right) =3\cdot 9-2\cdot 8=11,
\label{counting-3-2}
\end{equation}%
where $\tilde{I}_{m=3}$ denotes the number of $G\left( \mathbb{R}\right) $%
-invariant (electric resp. magnetic) polynomials constructed with 3 copies
of $\mathcal{R}_{el}$ resp. $\mathcal{R}_{magn}$, and we have used that the
generic $G\left( \mathbb{R}\right) $-orbit $\mathbf{O}_{m=3}$ spanned by 3
copies of $\mathcal{R}_{el}$ resp. $\mathcal{R}_{magn}$ has trivial identity
connected component of the stabilizer :
\begin{equation}
\mathbf{O}_{m=3}\simeq G(\mathbb{R})\Rightarrow \text{dim}_{\mathbb{R}%
}\left( \mathbf{O}_{m=3}\right) =\text{dim}_{\mathbb{R}}\left( G\left(
\mathbb{R}\right) \right) =2\cdot 8=16.
\end{equation}%
By taking $Q^{a}\in \mathcal{R}_{el}$, $Q_{a}\in \mathcal{R}_{magn}$, $%
Q^{a\alpha }\in \mathbb{R}^{3}\otimes \mathcal{R}_{el}$ and $Q_{a}^{\alpha
}\in \mathbb{R}^{3}\otimes \mathcal{R}_{magn}$, one can construct the $%
G\left( \mathbb{R}\right) $-invariants (electric and magnetic singlets)
\begin{eqnarray}
\mathbf{I}_{3,el} &:&=\frac{1}{3!}d_{abc}Q^{a}Q^{b}Q^{c}\equiv \mathbf{1}%
_{el}\mathbf{~}\text{of~}G\left( \mathbb{R}\right) ; \\
\mathbf{I}_{3,magn} &:&=\frac{1}{3!}d^{abc}Q_{a}Q_{b}Q_{c}\equiv \mathbf{1}%
_{magn}\text{~of~}G\left( \mathbb{R}\right) ,
\end{eqnarray}%
as well as the $\left( SL_{3,h}(\mathbb{R})\times G\left( \mathbb{R}\right)
\right) $-covariants%
\begin{eqnarray}
\mathbf{I}_{el}^{\alpha \beta \gamma } &:&=\frac{1}{3!}d_{abc}Q^{a\alpha
}Q^{b\beta }Q^{c\gamma }=\mathbf{I}_{el}^{(\alpha \beta \gamma )}\equiv S^{3}%
\mathbb{R}^{3}\equiv \left( \mathbf{10,1}_{el}\right) ~\text{of~}SL_{3,h}(%
\mathbb{R})\times G\left( \mathbb{R}\right) ;  \notag \\
&& \\
\mathbf{I}_{magn}^{\alpha \beta \gamma } &:&=\frac{1}{3!}d^{abc}Q_{a}^{%
\alpha }Q_{b}^{\beta }Q_{c}^{\gamma }=\mathbf{I}_{magn}^{(\alpha \beta
\gamma )}\equiv S^{3}\mathbb{R}^{3}\equiv \left( \mathbf{10,1}_{magn}\right)
~\text{of~}SL_{3,h}(\mathbb{R})\times G\left( \mathbb{R}\right) ,  \notag \\
&&
\end{eqnarray}%
where\footnote{%
From the explicit structure of $\mathcal{R}_{el}$ (resp. $\mathcal{R}_{magn}$%
), the Latin lowercase indices are actually realized as pairs of
contravariant-covariant (resp. covariant-contravariant) fundamental (Greek
lowercase) indices : $A^{a}\equiv A_{~\beta }^{\alpha }$ (resp. $A_{a}\equiv
A_{\alpha }^{~\beta }$). See e.g. \cite{FGimK} and \cite{Dasgupta-Wissanji}.}
$a,b,c=1,...,9$, and $\alpha ,\beta ,\gamma =1,2,3$ (and repeated indices are
summed over (Einstein's convention)). It is thus easy to
realize that the $\tilde{I}_{m=3}=11$ $G\left( \mathbb{R}\right) $-invariant
(electric resp. magnetic) polynomials constructed with 3 copies of $\mathcal{%
R}_{el}$ resp. $\mathcal{R}_{magn}$ organize as $\mathbf{10}\oplus \mathbf{1}
$ in terms of $SL_{3,h}(\mathbb{R})$-representations.\bigskip
\end{itemize}

In conclusion, the classification of the orbit structure of the
non-transitive action of%
\begin{equation}
\text{Sym}_{2}\ltimes \left( SL_{3,h}(\mathbb{R})\times \left( \SLL(3,%
\mathbb{R})^{\times 2}\right) \right) ~~\text{on~~}\left( \mathbf{3},\mathbf{%
3},\mathbf{3}^{\prime }\right) \text{~resp.~}\left( \mathbf{3},\mathbf{3}%
^{\prime },\mathbf{3}\right)
\end{equation}%
concerns the classification of the 3-centered extremal black holes resp.
black strings in the $\mathcal{N}=0$, $\mathfrak{J}_{3}^{\mathbb{C}_{s}}$%
-based Maxwell-Einstein-scalar theory \cite{MRR1, MRR2} in $D=5$.
Analogously, the orbit structure of the non-transitive action of%
\begin{equation}
\text{Sym}_{2}\ltimes \left( SL_{3,h}(\mathbb{R})\times \SLL(3,\mathbb{C})_{%
\mathbb{R}}\right) ~~\text{on~~}\left( \mathbf{3},\mathbf{3},\overline{%
\mathbf{3}}\right) \text{~resp.~}\left( \mathbf{3},\overline{\mathbf{3}},%
\mathbf{3}\right)
\end{equation}%
concerns the classification in the 3-centered extremal black holes resp.
black strings of the $\mathcal{N}=2$, $\mathfrak{J}_{3}^{\mathbb{C}}$-based
magic Maxwell-Einstein supergravity \cite{GST1, GST2} in $D=5$.

The exploitation of the classification of such orbits (which is the object
of this paper) for the study of 3-centered extremal black holes/strings in
the aforementioned Maxwell-Einstein (super)gravity theories in $D=5$
(Lorentzian) space-time dimensions goes beyond the scope of the present
investigation, and we leave it for further future work.

\section{Construction of the representations}\label{sec:E6}

In this section we show how the $\mathrm{SL}(3,\C)^3$-module $(\C^3)^{\otimes 3}$ can
be constructed using a $\Z/3\Z$-grading of the simple Lie algebra of type
$\mathfrak{e}_6$. This construction pertains to Vinberg's theory of
$\theta$-groups. 

Let $\g$ be a complex simple Lie algebra. Let $m\geq 2$ be an integer, and
fix a primitive $m$-th root of unity $\zeta$ in $\C$.
Then the $\Z/m\Z$-gradings
$$\g = \bigoplus_{i\in \Z/m\Z} \g_i$$
(where the $\g_i$ are subspaces such that $[\g_i,\g_j]\subset \g_{i+j}$)
correspond bijectively to the automorphisms of $\g$ of order $m$. Indeed,
if a grading as above is given then by defining $\theta(x) = \zeta^i x$ for
$x\in \g_i$ defines an automorphism of order $m$. Conversely, if $\theta$
is an automorphism of order $m$ then letting $\g_i$ be the eigenspace of
$\theta$ with eigenvalue $\omega^i$ yields a $\Z/m\Z$-grading. We also
remark that given a grading as above, the subalgebra $\g_0$ is reductive
(\cite[Lemma 8.1(c)]{kac}), and that $\g_1$ is naturally a $\g_0$-module. 

The finite order automorphisms have been classified, up to conjugacy in
$\Aut(\g)$, by several authors. One classification is due to Kac; we refer
to \cite[Chapter X, \S 5]{helgason} for a thorough treatment. The inner
automorphisms in this
classification correspond to labelings by non-negative integers of the affine
Dynkin diagram. Here we just consider the case where all labels
are 0 or 1. Let $\Phi$ be the root system of $\g$ with simple system $\Delta$.
Let $\tilde\alpha$ be the highest root of $\Phi$ with respect to the dominance
order induced by $\Delta$. Then the nodes of the affine Dynkin diagram
correspond to the elements of $\widetilde{\Delta}=\Delta\cup\{-\tilde\alpha\}$.
The weights of
the nodes $-\tilde\alpha$ is 1 whereas the weights of the other nodes
are the coefficients of $\tilde\alpha$ with respect to $\Delta$.
Let $\Pi_{0},\Pi_1\subset \widetilde{\Delta}$ be the subsets consisting of the
roots with label 0, 1 respectively. Then we have the following (see
\cite[Chapter 3, \S 3.7]{gov}):
\begin{itemize}
\item the automorphism $\theta$ associated to the labeling has order equal to
  the sum of the weights of the nodes corresponding to $\Pi_0$,
\item $\Pi_0$ generates the  root system of the semisimple part of $\g_0$,
\item the module $\g_1$ has $|\Pi_1|$ irreducible components.
\end{itemize}

Now let $G=\Aut(\g)^\circ$ be the identity component of the automorphism
group of $\g$. The Lie algebra of $G$ is isomorphic to $\g$ (to be precise,
it is equal to $\ad \g \subset \gl(\g)$). Let $G_0$ be the connected subgroup
of $G$ whose
Lie algebra is isomorphic to $\g_0$ (or, equal to $\ad \g_0$). Then $G_0$
naturally acts on $\g_1$. The representation of $G_0$ in
$\mathrm{GL}(\g_1)$ is called
a $\theta$-representation. They were introduced and studied in detail by
Vinberg, \cite{vinberg,vinberg2}. Among other things he developed methods
for classifying the orbits of $G_0$ in $\g_1$.

A first observation here is that $\g_1$ is closed under Jordan decomposition.
Indeed, if $x=x_s+x_n$ is the Jordan decomposition of $x\in \g_1$, then
$\theta(x) = \theta(x_s)+\theta(x_n)$ is the Jordan decomposition of
$\theta(x)$. But also $\theta(x) = \zeta x = \zeta x_s+\zeta x_n$ is the
Jordan decomposition of $\theta(x)$. Hence $x_s,x_n\in\g_1$. This immediately
divides the orbits into three groups:
nilpotent, semisimple and mixed orbits (that consists, respectively, of
nilpotent, semisimple, and mixed (that is, neither nilpotent nor semisimple)
elements). We will briefly indicate the methods for classifying these
orbits in the relevant sections of this paper.

Our main example is the automorphism of the simple Lie algebra
$\g$ of type $\mf{e}_6$ with Kac diagram displayed in Figure \ref{fig:e6}.

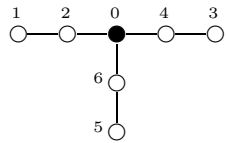
\begin{figure}[htb]\label{fig:e6}
\begin{center}
 \begin{picture}%\label{KacE6}
 (70,55)(0,0)
\setlength{\unitlength}{0.017in}
\put(35,18){\circle{5}} 
\put(35,3){\circle{5}}
\put(35,33){\circle*{5}}
\multiput(5,33)(15,0){5}{\circle{5}}
\multiput(8,33)(15,0){4}{\line(1,0){9}}\multiput(35,6)(0,15){2}{\line(0,1){9}}
\put(3,38){\tiny 1}
\put(18,38){\tiny 2}
\put(33,38){\tiny 0}
\put(48,38){\tiny 4}
\put(63,38){\tiny 3}
\put(28,3){\tiny 5}
\put(28,18){\tiny 6}
%\put(33,40){$1$}\put(39,1){$1$}
\end{picture}\end{center}
\caption{Kac diagram of an automorphism of $\mf{e}_6$}
\end{figure}

Here the black node has label 1 and the other nodes have label 0. From 
what is said above it follows that $\theta$ has order 3, $\g_0$ is
semisimple of type $3\mf{a}_2$, and the $\g_0$-module $\g_1$ is irreducible.

Let $\h$ be a Cartan subalgebra of $\g_0$. Then $\h$ is also a Cartan subalgebra
of $\g$. Let $\Phi$ be the corresponding root system of $\g$. For a root
$\alpha\in \Phi$ we denote the corresponding root space by $\g_\alpha$.
Let $\Delta=\{ \alpha_0, \alpha_1,\alpha_2,\alpha_3,\alpha_4,\alpha_6\}$ be a
set of simple roots of $\Phi$ such that in the corresponding Dynkin diagram
$\alpha_i$ corresponds to node $i$ as in Figure \ref{fig:e6}, and such that
$\g_{\alpha_i}\subset \g_0$ for $i=1,2,3,4,5,6$ and $\g_{\alpha_0}\subset \g_1$.
Accordingly we let $\alpha_5$ denote minus the highest root of $\Phi$.
For $0\leq i\leq 6$ we let $ x_i\in \g_{\alpha_i}$, $ h_i\in \h$,
$ y_i\in \g_{-\alpha_i}$ be ``canonical generators'', that is, letting $C$
denote the
Cartan matrix of the extended Dynkin diagram in Figure \ref{fig:e6}, we have
$$[ h_i, x_j]=C(j,i) x_j,\, [ h_i, y_j] = -C(j,i)  y_j,
[ x_i, y_j] = \delta_{ij}  h_i.$$

By \cite[\S 8]{vinberg}, $x_0$ is a lowest weight vector of the
$\g_0$-module $\g_1$. We have $[ h_i, x_0]=0$ for $i=1,3,5$ and
$[ h_i, x_0]=- x_0$ for $i=2,4,6$. Let $ v_0\in \g_1$ be a
highest weight vector. Then it follows that $[ h_i, v_0]= v_0$
for $i=1,3,5$ and $[ h_i, v_0]=0$ for $i=2,4,6$.

Now we let $\hat \g = \a_1\oplus \a_2\oplus \a_3$, where each $\a_i$ is a
copy of $\sll(3,\C)$. Let $\hat x_{i,j}, \hat h_{i,j}, \hat y_{i,j}$ for $j=1,2$
denote the canonical generators of $\a_i$. Concretely, this means the following.
Denote the $3\times 3$-matrix with a 1 on position $(i,j)$ and zeros
elsewhere by $e_{ij}$. Then
$$\hat x_{i,1} = e_{12},\, \hat x_{i,2} = e_{23},\, \hat h_{i,1} = e_{11}-e_{22},\,
\hat h_{i,2} = e_{22}-e_{33},\, \hat y_{i,1} = e_{21},\, \hat y_{i,2} = e_{32}.$$
Let $V^i$ for $1\leq i\leq 3$ be a copy of $\C^3$. In the sequel we denote
the standard basis of $\C^3$ by $e_0,e_1,e_2$.  Then $\hat \g$ acts on
$\cV=V^1\otimes V^2\otimes V^3$ by
$$(z_1+z_2+z_3) \cdot (v_1\otimes v_2\otimes v_3) =
z_1v_1\otimes v_2\otimes v_3 + v_1\otimes z_2v_2\otimes v_3 +
v_1\otimes v_2\otimes z_3v_3.$$

Then $\cV$ is an irreducible $\hat \g$-module with highest weight vector
$u_0=e_0\otimes e_0\otimes e_0$. We have $\hat h_{i,1} \cdot u_0 = u_0$ and
$\hat h_{i,2}\cdot u_0 = 0$.

Let $\psi : \hat\g \to \g_0$ be the isomorphism mapping $\hat a_{i,1} \mapsto
a_{2i-1}$, $\hat a_{i,2} \mapsto a_{2i}$, where $a\in \{ x,h,y\}$, $1\leq i\leq 3$.
Then $\g_1$ becomes a $\hat \g$-module by $\hat z \cdot x = [\psi(\hat z),x]$.
Comparing highest weights we see that the $\hat \g$-modules $\g_1$ and
$\cV$ are isomorphic, and that there is a unique $\hat \g$-module
isomorphism $\varphi : \g_1 \to \cV$ mapping $v_0\mapsto u_0$. 

Consider the group $\wG = \SLL(3,\C)\times \SLL(3,\C)\times \SLL(3,\C)$.
Then $\Lie(\wG) = \hat \g$. Let $G,G_0$ be as above, that is,
$G=\Aut(\g)^\circ$ and $G_0$ is the connected subgroup corresponding to the
subalgebra $\g_0$. Since $\wG$ is simply connected
we have a surjective homomorphism $\Psi : \wG \to G_0$ whose
differential is $\psi$. This defines a $\widehat G$-action on $\g_1$ by
$\hat g \cdot x = \Psi(\hat g) \cdot x$.
Also $\cV$ is a $\wG$-module by 
$$(g_1\times g_2\times g_3)\cdot (v_1\otimes v_2\otimes v_3) =
g_1v_1\otimes g_2v_2\otimes g_3v_3.$$
Then $\varphi$ is also an isomorphism of $\wG$-modules. It follows
that classifying the $\wG$-orbits in $\cV$ is equivalent to
classifying the $G_0$-orbits in $\g_1$.

\subsection{Permuting the tensor factors}\label{sec:perm}

Consider the vector space $\cV$. Let $\pi\in S_3$ be a permutation of
$\{1,2,3\}$. 
To $\pi$ we associate a linear map (denoted by the same symbol)
$\pi : \cV\to \cV$, $\pi(v_1\otimes v_2\otimes v_3 ) = v_{\pi(1)}\otimes
v_{\pi(2)} \otimes v_{\pi(3)}$. Furthermore we also define the map $\pi :
\wG\to \wG$, $\pi(g_1\times g_2\times g_3) =
g_{\pi(1)}\times g_{\pi(2)}\times g_{\pi(3)}$. Then obviously we have
$$\pi(g)\cdot \pi(v) = \pi(g\cdot v)\text{ for $g\in\wG$, $v\in V$.}$$
It follows that $\wG\cdot \pi(v) = \pi( \wG\cdot v)$. Furthermore,
denoting the stabilizer of $v\in \cV$ in $\wG$ by $Z_{\wG}(v)$ we have
$Z_{\wG}(\pi(v)) = \pi (Z_{\wG}(v))$.

As seen in Section \ref{sec:bhD5}
it is of interest to classify the orbits of $\wG$ in $V$ up to permutation of
the tensor factors, or up to permutation of two of the tensor factors.
In other words, instead of the group $\wG$ we consider the groups
$\mathrm{Sym}_3\ltimes \wG$ and $\mathrm{Sym}_2\ltimes \wG$ (where
$\mathrm{Sym}_2$ permutes the second and third factors only). For the
nilpotent elements the permutation action is given in detail in Table
\ref{CNilp}. For the semisimple and mixed elements see Remarks
\ref{rem:Symsemsim} and \ref{rem:Symnilp} respectively.

\subsection{The action on $V^1\otimes V^2 \otimes (V^3)^*$}\label{sec:333bar}

We can of course also consider the action of $\wG$ on the module
$V^1\otimes V^2 \otimes (V^3)^*$ (where the last factor is the dual module
of $V^3$). Here we argue that we get exactly the same orbits.

Let $\sigma_3 : \wG\to \wG$ be defined
by $\sigma_3(g_1,g_2,g_3) = (g_1,g_2,g_3^{-T})$. Then $\sigma_3$ is an
involution of $\wG$. Define the action $\circ$ of $\wG$ on $\cV$ by
$g\circ v = \sigma_3(g)\cdot v$. The $\wG$-module with this action
is isomorphic to $V^1\otimes V^2\otimes (V^3)^*$.
We have that $g\cdot v = \sigma_3(g)\circ v$. Hence
the orbits of $\wG$ on $\cV$ with respect to the $\cdot$ action, and with
respect to the $\circ$ action are the same. So in this paper we just consider
the $\cdot$ action.

\subsection{Notation}

Throughout the paper we use the \emph{bra and ket notation} (cf.
\cite[\S 5.4.2.1]{wallach}) for the elements of $\cV$, denoting the
elementary tensor $e_{j_1} \otimes e_{j_2} \otimes e_{j_3}$ 
by the symbol $|j_1\,j_2\,j_3 \rangle$.

By $\zeta$ we will denote a fixed third primitive root of unity in $\C$.

In the sequel we will freely use the notation introduced in this section.
In particular, we will use the groups $\wG$, $G$, $G_0$, and the Lie algebras
$\g$, $\g_0$.

\section{Galois cohomology and real orbits}\label{sec:galois}

One of the main goals of this paper is to achieve a classification of the
orbits of $\wG(\R) = \SLL(3,\R)\times \SLL(3,\R)\times \SLL(3,\R)$
on the real space spanned by the elementary tensors $|j_1\,j_2\,j_3 \rangle$
in $\cV$.
Our methods are based on Galois cohomology, to which we here give a
brief introduction.

Let $\sG$ be a subgroup of $\mathrm{GL}(n,\C)$, stable under complex conjugation
$g\mapsto \bar g$ (given by the complex conjugation of the matrix entries).
By $\sG(\R)$ we denote the group of real points,
$$\sG(\R) = \{ g\in \sG \mid \bar g = g\}.$$
We have the following definitions:
\begin{itemize}
\item a $g\in \sG$ is a {\em cocycle} if $\bar{g}g = 1$,
\item two cocycles $g_1,g_2$ are {\em equivalent} if there is an $h\in \sG$ such
  that $h^{-1} g_1 \bar{h} = g_2$,
\item the set of equivalence classes of cocycles is the {\em first Galois
  cohomology set} $\Ho^1 \sG$ of $\sG$.
\end{itemize}

We note that these definitions are an \emph{ad hoc} version of the classical
definition of Galois cohomology of a group, to the special case where the
Galois group has order 2. We refer to the book by Serre \cite{serre}
for a more general account. 

The crucial result for our computations is the following theorem; see
\cite[Section I.5.4]{serre}. 

\begin{theorem}\label{RealOrbits}
Suppose that $\Ho^1 \sG = 1$. Let $\sG$ act on the set $U$. Suppose that
$U$ has an involution $u\mapsto \bar u$ such that  $\overline{g\cdot u} =
\bar{g}\cdot \bar{u}$ for $g\in \sG$, $u\in U$. Let $u_0\in U$ be real (that is,
$\bar u_0 = u_0$) and consider the set of real points of the orbit of $u_0$,
$$\mc{O}(\R) = \{ v\in \sG\cdot u_0 \mid \bar v = v\}.$$
Let $Z_{\sG}(u_0) = \{ g\in \sG\mid g\cdot u_0=u_0\}$ be the stabilizer of $u_0$
in $\sG$. Then there exists a bijection between $\Ho^1 Z_\sG(u_0)$ and the set of
$\sG(\R)$-orbits in $\mc{O}(\R)$. This bijection is described explicitly
as follows. Let $[g]\in \Ho^1 Z_\sG(u_0)$, then there exists $h\in \sG$ such
that $h^{-1}\bar h = g$ (as  $\Ho^1 \sG = 1$). Then the class $[g]$ corresponds
to the $\sG(\R)$-orbit of $(h\cdot u_0)$
\end{theorem}

In order to use this we need to be able to compute the first Galois cohomology
set of matrix groups. This can be done by computer. However, we also use
some criteria to derive the cohomology sets directly. 

\begin{crit}\label{Hvan:torus}
We have that $\Ho^1 \mathrm{GL}(n,\C)$ and $\Ho^1 \SLL(n,\C)$ are both trivial.
So in particular $\Ho^1 T$ is trivial, where $T\subset \mathrm{GL}(n,\C)$ is a torus.  
\end{crit}

The next criterion is \cite[Corollary 3.2.2]{BorovoideGraafLe}.

\begin{crit}\label{Hvan:oddorder}
If $\sG$ is finite, abelian and of odd order then  $\Ho^1 \sG$ is trivial.
\end{crit}

\begin{crit}\label{crt:pgroup}
Let $\sG$ be a group of order $p^m$, where $p$ is an odd prime. Then
$\Ho^1 \sG$ is trivial.
\end{crit}

This follows from \cite[Lemma 3.2.3]{BorovoideGraafLe}.
The next criterion is \cite[Lemma 3.2.6]{BorovoideGraafLe}.

\begin{crit}\label{H:generators}
If $\sG$ is finite of order $2p^m$ where $p$ is an odd prime, 
then $\Ho^1 \sG=\{1, [c]\}$, where $c$ is a real element of order 2.
\end{crit}

The next criterion is \cite[Proposition 3.3.16]{BorovoideGraafLe}.

\begin{crit}\label{crit:alg1}
Let $\sG$ be a linear algebraic group with identity component $\sG^\circ$.
Suppose that $\sG/\sG^\circ$ is of order $p^n$ for some odd prime $p$ and
$n\geq 0$ and $|\Ho^1 \sG^\circ | <p$. Then the canonical map
$\Ho^1 \sG^\circ \to \Ho^1 \sG$ is bijective.
\end{crit}

\section{Nilpotent orbits}\label{sec:nil}

In this section we comment on the methods for classifying complex and
real nilpotent orbits. We present the complex classification of Nurmiev,
\cite{Nurmiev} to which we add some data (see below). We also classify the
real nilpotent orbits.

From \cite[\S 4.3]{BorovoideGraafLe} we recall some general results.
Let $\a = \oplus_{i\in \Z/m\Z} \a_i$ be a $\Z/m\Z$-graded semisimple Lie algebra
over a field of characteristic 0. Let $H$ be a group of automorpisms of $\a$
such that each element preserves the spaces $\a_i$ for $i\in \Z/m\Z$, and such
that $H$ contains $\exp \ad x$ for all nilpotent $x\in \a_0$. Then
\begin{itemize}
\item A nilpotent $e\in \a_1$ lies in a {\em homogeneous} $\sll_2$-triple
  $(h,e,f)$. This means that $h\in \a_0$, $f\in\a_{-1}$ and
  $$[h,e]=2e,\, [h,f]=-2f,\, [e,f]=h.$$
\item Two nilpotent elements $e,e'\in \a_0$, lying in homogeneous
  $\sll_2$-triples  $(h,e,f)$, $(h',e',f')$, are $H$-conjugate if and only if
  there is a $\sigma\in H$ with $\sigma(h)=h'$, $\sigma(e)=e'$, $\sigma(f)=f'$.
\end{itemize}  

\subsection{The complex nilpotent orbits}

Let $\a = \oplus_{i\in \Z/m\Z} \a_i$ be a $\Z/m\Z$-graded semisimple Lie algebra
over $\C$. Let $A_0$ be the connected subgroup of $\mathrm{Aut}(\a)$ with
Lie algebra $\ad_{\a} \a_0$. Then $A_0$ satisfies the requirements on $H$ above.
Moreover, as the base field is algebraically closed, it can be shown that if
we have two homogeneous
$\sll_2$-triples $(h,e,f)$, $(h',e',f')$ then $e,e'$ are $A_0$-conjugate if
and only if $h,h'$ are $A_0$-conjugate (cf., \cite[Theorem 8.3.6]{graaf17}).
Let $\h_0$ be a fixed Cartan subalgebra of $\a_0$. 
It follows that each nilpotent orbit has a representative $e$ lying in a
homogeneous $\sll_2$-triple $(h,e,f)$ such that $h\in \h_0$. We also have that
$h,h'\in\h_0$ are $A_0$-conjugate if and only if they are conjugate under
the Weyl group $N_{A_0}(\h_0)/Z_{A_0}(\h_0)$. This Weyl group is canonically
isomorphic to the Weyl group of the root system of $\a_0$ with respect to
$\h_0$. These facts can be used to devise an algorithm to classify the
nilpotent $A_0$-orbits in $\a_0$ (\cite[\S 8.4.1]{graaf17}). An alternative
is Vinberg's support method, \cite{vinberg2}, \cite[\S 8.4.2]{graaf17}.

In our case, that is, with $\g$, $\g_0$, $G$, $G_0$ as in Section
\ref{sec:E6}, there are 63 nonzero nilpotent orbits.
Nurmiev \cite{Nurmiev} lists them
up to permutation of the tensor factors (Section \ref{sec:perm}).
We have used the above indicated algorithms,
and their implementation in the {\sf SLA} (\cite{sla}) package of
{\sf GAP}4 (\cite{GAP4}), to check
the correctness of his list. We reproduce it here in Table \ref{CNilp}, where we
also add the permutations that can be used to determine representatives
of the orbits that are not in Nurmiev's list. 
The second column of this table has the same representatives as Nurmiev,
except that we use the qutrit notation. In the third column we indicate a
characteristic of the orbit by a list of integers. If $e$ is a representative
of a nilpotent orbit, lying in the homogeneous $\sll_2$-triple $(h,e,f)$,
then the element $h$ is called a characteristic of the orbit. The integers
in the third column are to be interpreted as
follows. Let $\alpha_1,\ldots,\alpha_6$ be a fixed set of simple roots of
$\g_0$ with respect to the fixed Cartan subalgebra $\h_0$. Then by the above
considerations it follows that each nilpotent orbit has a representative
$e$ lying in a homogeneous $\sll_2$-triple $(h,e,f)$ with $\alpha_i(h) \geq 0$
for $1\leq i\leq 6$. The values $\alpha_i(h)$ are given in the third column.
We remark that Nurmiev uses the action of the symmetric group $S_3$:
if two orbits are carried to each other via a permutation of the tensor factors,
then only one of them is listed. We do the same, except that in the fourth
column we list the permutations that have
to be applied to the listed orbits to obtain the other orbits. Thus the
total number of nilpotent orbits is 62 (excluding 0).
The fifth and sixth columns are devoted to the stablizer
$$Z_{\wG}(h,e,f) = \{ g\in \wG \mid g\cdot h=h,\,g\cdot e=e,\, g\cdot f=f\}$$
of a homogeneous $\sll_2$-triple $(h,e,f)$ containing a representative $e$ of
the given nilpotent orbit. The fifth column has a dscription of the
identity component, and the sixth column has a description of the
component group. These stablizers have been determined by computing explicit
sets of polynomial equations describing them, together with computational 
methods based on Gr\"obner bases. For this we have used the computer algebra
systems {\sf GAP}4 and {\sc Singular} (\cite{DGPS}). 

%\begin{landscape}
\begin{longtable}[l]{|l|c|c|c|c|c|}
\caption{Nilpotent Complex 3-qutrits}
\label{CNilp}\\
\hline
%\endfirsthead
% \multicolumn{5}{r}{\textit{(Continua alla pagina successiva)}}
% \endfoot
% \multicolumn{5}{l}{\textit{(Continua dalla pagina precedente)}}
% \endhead
% \hline
% \multicolumn{5}{c}{}\\
% \caption{Nilpotent 3-qutrits}
%\endlastfoot
 %  \begin{tabular}{|c|c|c|c|c|}
 %  \hline
$N$&3-qtrit&Char. & $\sigma$& $ Z_{\wG}(h,e)^\circ$ & $K$ \\
\hline
1&$\se{012}+\se{021}+\se{102}+\se{111}+\se{120}+\se{200}$&$\text{\tiny6\, 12\, 6\, 6\, 6\, 6} $&\tiny$\id$&\tiny$\id$&\tiny$(\Ztre)^2$ \\
&&$ \text{\tiny6\, 6\, 6\, 6\, 6\, 12 }$&\tiny$(12)$& &\\
&&$\text{\tiny 6\, 6\, 6\, 12\, 6\, 6 }$&\tiny$(132)$& &\\
\hline
2&$\se{012}+\se{021}+\se{102}+\se{110}+\se{111}+\se{200}$&$ \text{\tiny 6\, 6\, 0\, 6\, 6\, 6 } $&\tiny$\id$&\tiny$\id$&\tiny$(\Ztre)^2$ \\
&&$ \text{\tiny 6\, 6\, 6\, 6\, 0\, 6 }$&\tiny$(23)$& &\\
&&$ \text{\tiny   0\, 6\, 6\, 6\, 6\, 6 }$&\tiny$(123)$& &\\
\hline
3&$\se{002}+\se{011}+\se{020}+\se{101}+\se{112}+\se{200}$&$\text{\tiny 0\, 6\, 6\, 0\, 0\, 6 } $&\tiny$\id$&\tiny$\id$&\tiny $(\Ztre)^2 \times \Z / 2\Z$\\
&&$ \text{\tiny0\, 6\, 0\, 6\, 6\, 0 } $&\tiny$(23)$& &\\
&&$ \text{\tiny  6\, 0\, 0\, 6\, 0\, 6} $&\tiny$(123)$& &\\
\hline
4&$\se{002}+\se{011}+\se{101}+\se{110}+\se{220}$&$\text{\tiny 6\, 0\, 6\, 6\, 6\, 0 } $&\tiny$\id$&\tiny$T_1$&\tiny$\Ztre$ \\
&&$ \text{\tiny6\, 0\, 6\, 0\, 6\, 6 } $&\tiny$(23)$& &\\
&&$ \text{\tiny 6\, 6\, 6\, 0\, 6\, 0 }$&\tiny$(123)$& &\\
\hline
5&$\se{002}+\se{020}+\se{021}+\se{110}+\se{201}$&$ \text{\tiny3\, 3\, 0\, 6\, 3\, 3}  $&\tiny$\id$&\tiny$T_1$&\tiny$ \Ztre$ \\
&&$\text{\tiny 3\, 3\, 3\, 3\, 0\, 6  }$&\tiny$(23)$& &\\
&&$ \text{\tiny 0\, 6\, 3\, 3\, 3\, 3 }$&\tiny$(123)$&& \\
\hline
6&$\se{002}+\se{011}+\se{101}+\se{120}+\se{210}$&$ \text{\tiny 1\, 5\, 6\, 1\, 1\, 5} $&\tiny$\id$&\tiny$T_1 $&\tiny$\Ztre$ \\
&&$ \text{\tiny1\, 5\, 1\, 5\, 6\, 1 } $&\tiny$(23)$& &\\
&&$\text{\tiny 6\, 1\, 1\, 5\, 1\, 5 } $&\tiny$(123)$& &\\
\hline
7&$\se{002}+\se{011}+\se{020}+\se{101}+\se{210}$&$  \text{\tiny 3\, 0\, 3\, 3\, 3\, 3} $&\tiny$\id$&\tiny$T_1 $&\tiny$ \Ztre$ \\
&&$ \text{\tiny 3\, 3\, 3\, 3\, 3\, 0 }$&\tiny$(12)$& &\\
&&$ \text{\tiny 3\, 3\, 3\, 0\, 3\, 3 }$&\tiny$(132)$& &\\
\hline
8&$\se{002}+\se{020}+\se{111}+\se{200}$&$ \text{\tiny 2\, 4\, 2\, 4\, 2\, 4} $&\tiny$\id$&\tiny$ T_2 $ &\tiny $\id$\\
\hline
9&$\se{000}+\se{011}+\se{111}+\se{122}$&$ \text{\tiny0\, 6\, 0\, 0\, 0\, 0}  $&\tiny$\id$&\tiny$T_2 $&\tiny$S_3 \times \Ztre$ \\
&&$ \text{\tiny 0\, 0\, 0\, 0\, 0\, 6 }$&\tiny$(12)$& &\\
&&$  \text{\tiny 0\, 0\, 0\, 6\, 0\, 0 } $&\tiny$(132)$& &\\
\hline
10&$\se{002}+\se{011}+\se{020}+\se{101}+\se{110}+\se{200}$&$ \text{\tiny 2\, 2\, 2\, 2\, 2\, 2 }$&\tiny$\id$&\tiny$\id$&\tiny$(\Ztre)^3$ \\
\hline
11&$\se{002}+\se{020}+\se{101}+\se{210}$&$ \text{\tiny2\, 0\, 4\, 2\, 4\, 2 } $&\tiny$\id$&\tiny$ T_2 $&\tiny$\id$ \\
&&$ \text{\tiny 4\, 2\, 4\, 2\, 2\, 0 }$&\tiny$(12)$& &\\
&&$ \text{\tiny 4\, 2\, 2\, 0\, 4\, 2 }$&\tiny$(132)$& &\\
\hline
12&$\se{002}+\se{020}+\se{100}+\se{111}$&$ \text{\tiny1\, 5\, 0\, 1\, 0\, 1 } $&\tiny$\id$&\tiny$ T_2 $&\tiny$ \Ztre $ \\
&&$\text{\tiny  0\, 1\, 0\, 1\, 1\, 5 }$&\tiny$(12)$&& \\
&&$\text{\tiny 0\, 1\, 1\, 5\, 0\, 1}$&\tiny$(132)$&& \\
\hline
13&\tiny$\se{002}+\se{011}+\se{020}+\se{101}+\se{110}$&$ \text{\tiny 1\, 4\, 1\, 1\, 1\, 1}  $&\tiny$\id$&\tiny$ T_1 $&\tiny$(\Ztre)^2$ \\
&&$ \text{\tiny1\, 1\, 1\, 1\, 1\, 4  }$&\tiny$(12)$&& \\
&&$ \text{\tiny 1\, 1\, 1\, 4\, 1\, 1 }$&\tiny$(132)$& &\\
\hline
14&$\se{002}+\se{010}+\se{021}+\se{100}+\se{201}$&$ \text{\tiny 3\, 0\, 0\, 3\, 3\, 0} $&\tiny$\id$&\tiny$\SL $&\tiny$(\Ztre)^2$ \\
&&$ \text{\tiny3\, 0\, 3\, 0\, 0\, 3}  $&\tiny$(23)$& &\\
&&$ \text{\tiny 0\, 3\, 3\, 0\, 3\, 0 }$&\tiny$(123)$& &\\
\hline
15&$\se{011}+\se{022}+\se{100}$&$ \text{\tiny 2\, 4\, 1\, 0\, 1\, 0 }$&\tiny$\id$&\tiny$\GL \times T_1$&\tiny$\id$ \\
&&$ \text{\tiny1\, 0\, 1\, 0\, 2\, 4 } $&\tiny$(12)$& &\\
&&$  \text{\tiny1\, 0\, 2\, 4\, 1\, 0 }$&\tiny$(132)$& &\\
\hline
16&$\se{002}+\se{011}+\se{020}+\se{100}$&$ \text{\tiny 2\, 2\, 1\, 1\, 1\, 1}  $&\tiny$\id$&\tiny$T_2 $&\tiny$ \Ztre$ \\
&&$ \text{\tiny 1\, 1\, 1\, 1\, 2\, 2 }$&\tiny$(12)$& &\\
&&$ \text{\tiny  1\, 1\, 2\, 2\, 1\, 1} $&\tiny$(132)$& &\\
\hline
17&$\se{001}+\se{010}+\se{102}+\se{120}$&$ \text{\tiny 0\, 4\, 2\, 0\, 2\, 0} $&$\id$&\tiny$\GL $&\tiny$\Ztre$ \\
&&$ \text{\tiny 2\, 0\, 2\, 0\, 0\, 4 }$&\tiny$(12)$& &\\
&&$ \text{\tiny 2\, 0\, 0\, 4\, 2\, 0 }$&\tiny$(132)$& &\\
\hline
18&$\se{000}+\se{011}+\se{101}+\se{112}$&$ \text{\tiny0\, 3\, 0\, 0\, 0\, 3 } $&\tiny$\id$&\tiny$\GL$&\tiny$\Ztre$ \\
&&$\text{\tiny 0\, 3\, 0\, 3\, 0\, 0 } $&\tiny$(23)$& &\\
&&$ \text{\tiny 0\, 0\, 0\, 3\, 0\, 3 }$&\tiny$(123)$& &\\
\hline
19&$\se{002}+\se{010}+\se{101}$&$\text{\tiny 1\, 2\, 0\, 1\, 1\, 2 } $&\tiny$\id$&\tiny $T_3$&\tiny$\id$\\
&&$ \text{\tiny  1\, 2\, 1\, 2\, 0\, 1} $&\tiny$(23)$& &\\
&&$ \text{\tiny 0\, 1\, 1\, 2\, 1\, 2 }$&\tiny$(123)$& &\\
\hline
20&$\se{000}+\se{111}$&$\text{\tiny 0\, 2\, 0\, 2\, 0\, 2 } $&\tiny$\id$&\tiny $T_4$&\tiny$ \Z / 2\Z$\\
\hline
21&$\se{001}+\se{010}+\se{100}$&$ \text{\tiny 1\, 1\, 1\, 1\, 1\, 1} $&\tiny$\id$&\tiny$T_3$&\tiny$\id$ \\
\hline
22&$\se{000}+\se{011}+\se{022}$&$ \text{\tiny3\, 0\, 0\, 0\, 0\, 0 } $&\tiny$\id$&\tiny $\SL \times \SLL(3,\C) $&\tiny$(\Ztre)^2$\\
&&$ \text{\tiny0\, 0\, 0\, 0\, 3\, 0 } $&\tiny$(12)$& &\\
&&$ \text{\tiny 0\, 0\, 3\, 0\, 0\, 0 }$&\tiny$(132)$& &\\
\hline
23&$\se{000}+\se{011}$&$\text{\tiny 2\, 0\, 0\, 1\, 0\, 1}  $&\tiny$\id$&\tiny$(\GL)^2$&\tiny$\id$ \\
&&$ \text{\tiny 0\, 1\, 0\, 1\, 2\, 0} $&\tiny$(12)$& &\\
&&$  \text{\tiny0\, 1\, 2\, 0\, 0\, 1} $&\tiny$(132)$& &\\
\hline
24&$\se{000}$&$ \text{\tiny1\, 0\, 1\, 0\, 1\, 0 } $&\tiny$\id$&\tiny $(\GL)^2 \times \SL$ &\tiny$\id$\\
\hline
% \end{tabular}
\end{longtable}
%\end{landscape}

\subsection{The real nilpotent orbits}

Let $T$ denote the set of complex homogeneous $\sll_2$-triples in $\g$.
Then $\wG$ acts on $T$. Furthermore, $T$ is closed under complex
conjugation,
and the triples fixed under complex conjugation form the set $T^\R$ of real
homogeneous $\sll_2$-triples. The group $\wG(\R)$ acts on $T^\R$ and
as seen
at the beginning of the section, the nilpotent $\wG(\R)$-orbits are in
bijection with the $\wG(\R)$-orbits in $T^\R$. Since
$\Ho^1 \wG=1$, Theorem \ref{RealOrbits} implies the following.

\begin{theorem}\label{NilpCent}
Let $e \in \g_1$ be a nilpotent element and $(h,e,f)$ an associated
homogeneous $\sll_2$-triple. Let $\mathcal{O}=\wG\cdot e$ and
$\mathcal{O}(\R) = \{ x\in \mathcal{O}\mid \bar x =x\}$.
There is a bijection between
$\wG(\R)$-orbits in $\mathcal{O}(\R)$ and $\Ho^1(Z_{\wG}(h,e,f))$.
\end{theorem}

Among the centralizers in Table \ref{CNilp}, only three cases have non
trivial Galois cohomology. In order to compute generators of first cohomology
set, we need a more explicit description of their stabilizers, which in
all cases we denote $Z$. 

\begin{enumerate}
\item[3] $\e{002}+\e{011}+\e{020}+\e{101}+\e{112}+\e{200}$: the stabilizer is isomorphic to $\Ztre \times \Ztre \times \Z / 2\Z$ and consists of the elements
  $$\diag(\delta^2\zeta^2,\epsilon\delta^2\zeta^2,\epsilon\delta^2\zeta^2)
  \times \diag(\delta,\epsilon\delta,\epsilon\delta) \times
  \diag(\epsilon\zeta,\epsilon\zeta,\zeta),$$
  where $\zeta^3=\delta^3=1$ and $\epsilon^2=1$. By Criterion \ref{H:generators}
  it follows that the $\Ho^1 Z$ consists of $[1]$ and $[c]$ where
  $$c=\diag(1,-1,-1)\times\diag(1,-1,-1)\times\diag(-1,-1,1).$$
  An element $g\in \mathrm{SL}_3(\C)^3$ with $g^{-1} \bar{g} = c$ is
  $$g = \diag(-1,i,i)\times \diag(-1,i,i)\times \diag(i,i,-1).$$
  The corresponding orbit representative is
  $$-\e{002}+\e{011}+\e{020}+\e{101}+\e{112}+\e{200}.$$
\item[9]  $\e{000}+\e{011}+\e{111}+\e{122}$: the identity component of the
  stabilizer is a 2-dimensional torus, and the component group is isomorphic
  to $S_3 \times \Ztre$. The identity component $T_2$ consists of the
  elements $\diag(1,1,1)\times \diag(s^{-1},t^{-1},st)\times\diag(s, t,
  s^{-1}t^{-1})$. The component group is generated by $\diag(\zeta^2,\zeta^2,
  \zeta^2) \times \diag(\zeta,\zeta,\zeta) \times \diag(1,1,1)$ and
  $$g_1 = \begin{pmatrix} -1&0&0\\ -1&1&0\\ 0&0&-1 \end{pmatrix}\times
  \begin{pmatrix} 0&-1&0\\ -1&0&0\\ 0&0&-1\end{pmatrix}\times\\
    \begin{pmatrix} 0&1&0\\ 1&0&0\\ 0&0&-1\end{pmatrix},$$
  $$g_2 =  \begin{pmatrix} 0&-1&0&\\ 1&-1&0\\ 0&0&1 \end{pmatrix}\times
  \begin{pmatrix} 0&-1&0\\ 0&0&-1\\ 1&0&0\end{pmatrix}\times\\
    \begin{pmatrix} 0&1&0\\ 0&0&1\\ 1&0&0\end{pmatrix}.$$
      We have $g_1^2=1$, $g_2^3=1$ and $g_1g_2g_1 = g_2^2 \bmod T_2$. So $g_1,g_2$
      generate a subgroup of the component group isomorphic to $S_3$.
      Using \cite[Proposition 3.3.17]{BorovoideGraafLe} it can be shown that
      $\Ho^1 Z = \{ [1], [g_1] \}$.
      An element $g\in \mathrm{SL}_3(\C)^3$ with $g^{-1} \bar{g} = g_1$ is
      $$\begin{pmatrix} \tfrac{1}{2}&-1&0&\\ 0&0&-i\\ -i&0&0 \end{pmatrix}\times
  \begin{pmatrix} 1&-1&0\\ i&i&0\\ 0&0&-\tfrac{1}{2}i\end{pmatrix}\times \\
    \begin{pmatrix} 1&1&0\\ i&-i&0\\ 0&0&\tfrac{1}{2}i\end{pmatrix}.$$
      The corresponding orbit representative is
      $$2\e{210}+2\e{201}-\tfrac{1}{4}\e{022}-\e{011}+\e{000}.$$
\item[20] $\e{000}+\e{111}$: the stabilizer is isomorphic to $T_4 \rtimes \Z / 2\Z$.
  The identity component $T_4$ consists of
  $$\diag(s^{-1}u^{-1}, t^{-1}v^{-1}, stuv )\times \diag(s,t,s^{-1}t^{-1}) \times
  \diag(u,v,u^{-1}v^{-1}).$$
  The component group is generated by
  $$g_0 = \begin{pmatrix} 0&1&0\\ 1&0&0\\ 0&0&-1\end{pmatrix} \times
    \begin{pmatrix} 0&1&0\\ 1&0&0\\ 0&0&-1\end{pmatrix}\times
      \begin{pmatrix} 0&1&0\\ 1&0&0\\ 0&0&-1\end{pmatrix}.$$
        We have $\Ho^1 Z = \{ [1], [g_0] \}$.
        An element $g\in \mathrm{SL}_3(\C)^3$ with $g^{-1} \bar{g} = g_0$ is
        $$\begin{pmatrix} 1&1&0\\ i&-i&0\\ 0&0&\tfrac{1}{2}i\end{pmatrix}\times
          \begin{pmatrix} 1&1&0\\ i&-i&0\\ 0&0&\tfrac{1}{2}i\end{pmatrix}\times
          \begin{pmatrix} 1&1&0\\ i&-i&0\\ 0&0&\tfrac{1}{2}i\end{pmatrix}.$$
   The corresponding orbit representative is
   $$2\e{000}-2\e{011}-2\e{101}-2\e{110}.$$
\end{enumerate}

Summarizing, apart from the real nilpotent orbits with representatives
given in Table \ref{CNilp} we get the real nilpotent orbits with representatives
given in Table \ref{RNilp}.
\begin{center}
\begin{longtable}[c]{|l|c|c|}
\caption{Nilpotent real 3-qutrits}
\label{RNilp}\\
\hline
$N$&3-qutrit & $\sigma$ \\
\hline
3 & $-\e{002}+\e{011}+\e{020}+\e{101}+\e{112}+\e{200}$ & $\id$\\
& & $(23)$ \\
& & $(123)$\\
\hline
9 & $2\e{210}+2\e{201}-\tfrac{1}{4}\e{022}-\e{011}+\e{000}$ & $\id$ \\
& & $(12)$ \\
& & $(132)$ \\
\hline
20 & $2\e{000}-2\e{011}-2\e{101}-2\e{110}$ & $\id$ \\
\hline
\end{longtable}
\end{center}

\section{Semisimple Elements}\label{semis}

In this section we consider the orbits of semisimple elements of $\g_1$.
We first describe the complex classification where we follow Nurmiev
\cite{Nurmiev} adding many details on the proof of the correctness of the
classification. These are then also needed for the classification in the
real case.

\subsection{Semisimple orbits: the complex case}\label{sec:semsimC}

Much of the theory needed to classify the semisimple orbits in $\g_1$ is due
to Vinberg \cite{vinberg}. Here we give a short overview of some of the main
concepts and results. 

A Cartan subspace $\mf{C}\subset \cV$ is by definition a maximal subspace
consisting of commuting semisimple elements. Vinberg \cite{vinberg} showed
that any two Cartan subspaces are conjugated with respect to $G_0$.
As a consequence, every semisimple orbit has a representative in any given
Cartan subspace.
\begin{defn}
The Weyl group of a Cartan subspace $\mf{C}$ is $W_\gC=N_{G_0}(\mf{C})/Z_{G_0}(\mf{C})$.
\end{defn}

The group $W_\gC$ is a finite complex reflection group,
\cite[Theorem 8]{vinberg}. Moreover, two
elements from $\mf{C}$ are $G_0$-conjugate if and only if they are
$W_\gC$-conjugate, \cite[Theorem 2]{vinberg}. So the classification of
of the semisimple orbits reduces to the classification of the $W_\gC$-orbits
in $\mf{C}$ where the latter is any fixed Cartan subspace. Those orbits
are infinitely many, and it does not seem to be possible to present them
in an irredundant list. However, it is possible to define several special
subsets of $\mf{C}$ such that each semisimple orbit has a point in exactly one
of them. Secondly, two elements of the same subset are $G_0$-conjugate
if and only if they are conjugate under an explicitly given finite group.
Here we show how that is done in our case.

Consider the following elements
\begin{align*}
u_1 &=\e{000}+\e{111}+\e{222}\\
u_2 &=\e{012}+\e{120}+\e{201}\\
u_3 &=\e{021}+\e{210}+\e{102}.
\end{align*}

By \cite{Nurmiev}, see also \cite[\S 5.4.4, Exercise 2]{wallach}, these
span a Cartan subspace $\gC$ of $\g_1$. We also checked this by computer.
In the sequel we fix this space.

Now we briefly review the construction of complex reflections in the Weyl
group given in \cite[\S 3.3]{VinEl}, adapted to our situation. Let $\h$ be
the centralizer of $\gC$ in the Lie algebra $\g$. Then $\h$ is
a Cartan subalgebra of $\g$. Let $\Phi$ be the root system of
$\g$ with respect to $\h$. Consider the automorphism $\theta$ of
order 3 (see Section \ref{sec:E6}). This automorphism stabilizes $\h$.
It induces a map $\theta_*$ of the dual space $\h^*$ by $\theta_*(\mu)(h) =
\mu(\theta(h))$ for $h\in \h$. For a root $\alpha$ of $\h^*$ we have that
the elements $\pm \alpha$, $\pm \theta_*(\alpha)$, $\pm \theta_*^2(\alpha)$
span a root subsystem of type $A_2$. By $U(\alpha)$ we denote the corresponding
subalgebra of $\g$. Let $\gC(\alpha) = \gC \cap U(\alpha)$, then
$\gC(\alpha)$ is of dimension 1. Let $\gC_0(\alpha) = \{ p\in \gC \mid
\alpha(p)=0\}$. Then $\gC = \gC_0(\alpha) \oplus \gC(\alpha)$. Moreover,
defining the linear map $w_\alpha : \gC\to \gC$ by $w_\alpha(p) = p$ for
$p\in \gC_0(\alpha)$, $w_\alpha(p) = \zeta p$ for $p\in \gC(\alpha)$ we obtain a
complex reflection of $\gC$. The arguments in \cite[\S 3.3]{VinEl} can
be used also here to show that $w_\alpha$ is induced by an element of
$N_{G_0}(\gC)$. Hence $w_\alpha \in W_\gC$.

The root system $\Phi$ is partitioned into twelve root subsystems of type $A_2$
as constructed above. This yields 12 reflections in $W_\gC$. In fact, they form
a single conjugacy class. The reflections of the form $w_\alpha^2$ also
form a conjugacy class, and these are the two conjugacy classes in $W_\gC$
consisting of complex reflections. 

In this case the Weyl group $W_\gC$ has order $648$. By the table in
\cite[\S 9]{vinberg} it is isomorphic to the
group $G_{25}$ in the Shephard-Todd classification, and is also denoted
$W(\mathcal{L}_3)$ in \cite{lehta}. 
Using the above
techniques we can find all reflections in $W_\gC$ and show by a straightforward
{\sf GAP} calculation that $W_\gC$ is generated by the complex reflections:

\[\begin{pmatrix} \zeta&0&0\\0&1&0\\0&0&1 \end{pmatrix}, \; \begin{pmatrix} 1&0&0\\0&\zeta&0\\0&0&1 \end{pmatrix}, \; \frac{1}{3}\begin{pmatrix} 2+\zeta&-1-2\zeta&2+\zeta\\2+\zeta&2+\zeta&-1-2\zeta\\-1-2\zeta&2+\zeta&2+\zeta \end{pmatrix}.\]

Now we describe how to obtain the special subsets of $\gC$ mentioned above.
By $W_p$ we denote the stabilizer of $p\in \gC$ in $W_\gC$. It is known
that also $W_p$ is generated by complex reflections,
\cite[Theorem 1.5]{steinberg}. 
For $p\in \gC$ set

\begin{align*}
\gC_p&=\{h \in \gC \, | \, w \, h = h  \, \text{ for all } w \in W_p\}\\
\gC_p^\circ &= \{ q \in \gC_p \, |  W_q=W_p\}.
\end{align*}

\begin{rmk}\label{Czarop}
The set $\gC_p^\circ $ is Zariski open in $\gC_p$. It is determined by the
inequalities  $w \, q \neq q$ for all
$w \in W_\gC \setminus W_p$. Because also the stabilizer $W_q$ is generated by
complex reflections we have that $q\in \gC_p$ lies in $\gC_p^\circ$ if and
only if $w\, q \neq q$ for all complex reflections in $W_\gC$ that do not lie in
$W_p$. Let $w$ be such a reflection. Then $w=w_\alpha$ or $w=w_\alpha^2$ for an
$\alpha\in \Phi$. From the description of $w_\alpha$ it follows that
$w(q)\neq q$ is equivalent to $\alpha(q)\neq 0$. We can compute a linear
polynomial $\psi_\alpha=c_1x_1+c_2x_2+c_3x_3$ such that $\alpha(a_1u_1+a_2u_2
+a_3u_3) = \psi_\alpha(a_1,a_2,a_3)$. Hence for $q=a_1u_1+a_2u_2+a_3u_3$ we
have that $w(q)\neq q$ if and only if $\psi_\alpha(a_1,a_2,a_3)\neq 0$. So by
taking the product of all $\psi_\alpha$ where $\alpha$ is such that $w_\alpha
\not\in W_p$ we can compute a polynomial inequality defining the set
$\gC_p^\circ$.
\end{rmk}

Now for $v \in W_\gC$ and $p,q\in \gC$ it is immediate that
\begin{equation}\label{eq:conj}
v\cdot \gC_p^\circ = \gC_q^\circ \text{  if and only if  } W_q = vW_pv^{-1}.
\end{equation}
Let $R$ denote the collection of all reflection subgroups of $W_\gC$ that are equal to a $W_p$ for a 
$p\in \gC$. Then it follows that $R$ is closed under conjugacy by $W_\gC$. Let 
$p_1,\ldots,p_m\in \gC$ be such that $W_{p_1},\ldots,W_{p_m}$ are representatives of the different
conjugacy classes in $R$. Then \eqref{eq:conj} implies that each semisimple orbit has a point in a
unique $\gC_{p_i}^\circ$. 

Let $p$ be one of the $p_i$. By \eqref{eq:conj} we see that the group of elements of $W_\gC$ mapping
$\gC_p^\circ$ to itself is exactly the normalizer $N_{W_\gC}(W_p)$. By definition of $\gC_p^\circ$
it follows that the group of elements fixing each element of $\gC_p^\circ$ is exactly $W_p$.
Hence the group $\Gamma_p = N_{W_\gC}(W_p) / W_p$ acts naturally on $\gC_p^\circ$.
Again from \eqref{eq:conj} we immediately get the following result.

\begin{theorem} Two elements of $\gC_p^\circ$ are $G_0$-conjugate if and only if they are $\Gamma_p$-conjugate.
\end{theorem}

We also have the following lemma, which is of fundamental importance for
the complex and for the real classification.

\begin{lemma}\label{lem:stab}
Let $p\in \gC$ and $q_1,q_2\in \gC_p^\circ$. Let $Z_{G_0}(q_i) = \{ g\in G_0 \mid
g\cdot q_i = q_i\}$ be their stabilizers in $G_0$. Then $Z_{G_0}(q_1) =
Z_{G_0}(q_2)$.  
\end{lemma}

\begin{proof}
Let $w\in W_\gC$ be a complex reflection. Then in the above notation we have
$w=w_\alpha$ or $w=w_\alpha^2$ for an $\alpha\in \Phi$. By the construction
of the $w_\alpha$ it follows that $w_\alpha(q) = q$ if and only if $\alpha(q)=0$.
Set $\Phi_i = \{ \alpha\in \Phi \mid \alpha(q_i) = 0\}$. Then as $W_{q_1} =
W_{q_2}$ it follows that $\Phi_1=\Phi_2$. Let $\mf{z}_i = \{ x \in \g \mid
[x,q_i]=0\}$. Then
$$\mf{z}_i = \h \bigoplus_{\alpha\in \Phi_i} \g_\alpha,$$
where $\g_\alpha$ denotes the root space in $\g$ corresponding to the
root $\alpha$. It follows that $\mf{z}_1 = \mf{z}_2$. By
\cite[Corollary 3.11]{steinberg2}
the stablizers $Z_{G}(q_i)$ are connected. The Lie algebra of $Z_{G}(q_i)$ is
$\mf{z}_i$. So $Z_G(q_1) = Z_G(q_2)$. But $Z_{G_0}(q_i) = G_0\cap Z_G(q_i)$.
\end{proof}

As seen above, each semisimple element is conjugate to one of the form
$u=a_1u_1+a_2u_2+a_3u_3$.
In \cite{Nurmiev} the representatives for semisimple orbits are classified in terms of the coefficients $\{a_i\}$ in five families, corresponding to the subsets of the form $\gC_p^\circ$. Each family corresponds to a conjugacy class of
reflection subgroups of $W_\gC$. It is known that $W_\gC$ has six conjugacy
classes of reflection subgroups;  see e.g., \cite[Table 3]{taylor}, where
the subgroups apart from the trivial subgroup and $W_\gC$ itself are denoted
$l_1$, $2l_1$, $l_2$ and $3l_1$. The group $3l_1$ is not the stabilizer of a
point, as it is of rank 3 and just stabilizes 0. So when we include the trivial
subgroup and $W_\gC$ itself we get five families of the form $\gC_p^\circ$.
Here we give the description of these families following \cite{Nurmiev}.
The fifth family is omitted because it is constituted only by the null vector.
This family corresponds to the subgroup which is $W_\gC$ itself. By Lemma
\ref{lem:stab} the elements of a fixed family all have the same stabilizer
in $G_0$, and hence in $\wG$. 
We also describe the stabilizers in $\wG$. They have been determined using
computational techniques based on Gr\"obner bases (as described in
\cite{BorovoideGraafLe}). We also
give an explicit description of the groups $\Gamma_p$. 

\textbf{First Family} This family corresponds to the trivial subgroup. 
The parameters $a_i$ satisfy the conditions \[a_1 a_2 a_3 \neq 0\]
\[\left(a_1^3+a_2^3+a_3^3\right)^3 - \left(3a_1 a_2 a_3 \right)^3 \neq 0\]
The stabilizer of the semisimple elements in the first family is finite of order 81 and 
generated by 
\begin{align*}
    &\diag(\zeta,\zeta^2,1,\zeta,\zeta^2,1,\zeta,\zeta^2,1)\\
    &\diag(\zeta^2,\zeta^2,\zeta^2,1,1,1,\zeta,\zeta,\zeta)\\
    &\begin{pmatrix} 0&0&\zeta^2\\1&0&0\\0&\zeta&0 \end{pmatrix}\times
    \begin{pmatrix}0&0&\zeta\\ \zeta^2&0&0\\ 0&1&0\end{pmatrix}\times
    \begin{pmatrix}0&0&1\\ \zeta&0&0\\ 0&\zeta^2&0\end{pmatrix}.
\end{align*}

Here the group $\Gamma_p$ is $W_\gC$.

\textbf{Second Family} This family corresponds to the reflection subgroup of order 3 generated
by $\diag(1,1,\zeta)$. It consists of the elements $a_1u_1+a_2u_2$. The parameters $a_1,a_2$ satisfy the open condition \[a_1 a_2 \left(a_1^3 + a_2^3\right) \neq 0.\]
The stabilizer of the elements of the second family is of the form $C\ltimes T_2$, where
$T_2$ is a 2-dimensional torus consisting of elements of the form
$$T_2(t_1,t_2) = \diag(t_1^{-1}t_2^{-1},t_1,t_2, t_1,t_2, t_1^{-1}t_2^{-1},t_2,t_1^{-1}t_2^{-1},t_1),$$
and $C$ is a group of order 9 generated by 
\begin{align*}
&\diag(\zeta,\zeta,\zeta,\zeta^2,\zeta^2,\zeta^2,1,1,1)\\
&\begin{pmatrix} 0&0&1\\1&0&0\\0&1&0 \end{pmatrix}\times
    \begin{pmatrix}0&0&1\\ 1&0&0\\ 0&1&0\end{pmatrix}\times
    \begin{pmatrix}0&0&1\\ 1&0&0\\ 0&1&0\end{pmatrix}
\end{align*}    

The group $\Gamma_p$ is of order 18 and generated (with respect to the basis
$u_1,u_2$) by 
$$\begin{pmatrix} 0 & -1 \\ -1 & 0 \end{pmatrix}, ~ \begin{pmatrix} 1 & 0 \\ 0 & \zeta\end{pmatrix}.$$

\textbf{Third Family} This family corresponds to the reflection subgroup of order 9 generated by 
$\diag(1,1,\zeta)$ and $\diag(1,\zeta,1)$. It consists of the elements $a_1 u_1$ with $a_1 \neq 0$.
The stabilizer of the elements of the third family is of the form $C\ltimes T_4$, where
$T_4$ is a 4-dimensional torus consisting of elements of the form
$$T_4(t_1,t_2,t_3,t_4)=\diag(t_1^{-1}t_3^{-1},t_2^{-1}t_4^{-1},(t_1t_2t_3t_4)^{-1},t_1,t_2,(t_1t_2)^{-1},t_3,t_4,(t_3t_4)^{-1}).$$
and $C$ is a group of order 3 generated by 
$$\begin{pmatrix} 0&0&1\\1&0&0\\0&1&0 \end{pmatrix}\times
    \begin{pmatrix} 0&0&1\\1&0&0\\0&1&0 \end{pmatrix}\times
    \begin{pmatrix} 0&0&1\\1&0&0\\0&1&0 \end{pmatrix}.$$

The group $\Gamma_p$ is of order 6 and generated by $-\zeta$.

\textbf{Fourth Family} This family corresponds to the reflection subgroup of order 24 generated by
$\diag(\zeta,1,1)$ and 
$$\tfrac{1}{3}\begin{pmatrix} 2+\zeta & 2+\zeta & 2+\zeta \\ -1-2\zeta & 2+\zeta & -1+\zeta\\
-1-2\zeta & -1+\zeta & 2+\zeta \end{pmatrix}.$$
It consists of the elements $a(u_2-u_3)$.
The stabilizer of the elements of the fourth family is of the form $F\ltimes C^\circ$, where
$C^\circ$ consists of all matrices $A\times A\times A$, with $A\in \SLL(3,\C)$
and $F$ is a group of order 3 generated by
$\diag(\zeta,\zeta,\zeta,\zeta^2,\zeta^2,\zeta^2,1,1,1)$.

The group $\Gamma_p$ is generated by $\zeta$.

\begin{rmk}\label{rem:Symsemsim}
Here we comment on the classification up to permutation of the tensor factors
(see Section \ref{sec:perm}). The elements of the families 1 and 3 are stable
under all permutations. The element $a(u_2-u_3)$ of the fourth family
is by a permutation of
order 2 mapped to $-a(u_2-u_3)$, whereas it is stable under permutations of
order 3. So if we consider the action of $\mathrm{Sym}_2$, $\mathrm{Sym}_3$ then
$a(u_2-u_3)$ is conjugate to $-a(u_2-u_3)$. So here we do not consider the
group $\Gamma_p$ of order 3, but rather a group of order 6.
The elements $a_1u_1+a_2u_2$ of the
second family are mapped under a permutation of order 2 to $a_1u_1+a_2u_3$.
The Weyl group $W_\gC$ contains the transformation
$$\begin{pmatrix} 0&0&1\\ 1&0&0\\ 0&1&0\end{pmatrix}$$
mapping $a_1u_1+a_2u_3$ to $a_2u_1+a_1u_2$. So if we also consider the action of
the symmetric group, $a_1u_1+a_2u_2$ is conjugate to $a_2u_1+a_1u_2$. So in this
case the elements of the second family are conjugate under
$\mathrm{Sym}_k\ltimes\wG$ ($k=2,3$) if and only if they are conjugate under
a group of order 36. 
\end{rmk}  

\subsection{Semisimple orbits: the real case}\label{semsim:real}

Now we recall some constructions from \cite{BorovoideGraafLe2}. The proofs
in the latter paper were relative to the special case considered in that
paper. However, because of Lemma \ref{lem:stab} all proofs go through also in
our case. Therefore we omit them.

Let $p\in \gC$ and write $\cF = \gC_p^\circ$. Recall that $\Gamma_p =
N_{W_{\gC}}(W_p)/W_p$. Let
\begin{align*}
  N_{\wG}(\cF) &= \{ g\in \wG \mid g\cdot q \in \cF \text{ for all } q\in \cF\}\\
  Z_{\wG}(\cF) &= \{ g\in \wG \mid g\cdot q=q \text{ for all } q\in \cF\}.
\end{align*}  
Define a map $\varphi : N_{\wG}(\cF)\to \Gamma_p$ in the following way. Let
$g\in N_{\wG}(\cF)$; then $g\cdot p\in \cF$ and hence there is a $w\in
N_{W_{\gC}}(W_p)$ with $g\cdot p = w\cdot p$; and we set $\varphi(g) = wW_p$.
Then $\varphi$ is well-defined and a surjective group homomorphism with kernel
$Z_{\wG}(\cF)$.

We have the following theorem. For a proof see
\cite[Proposition 5.2.4]{BorovoideGraafLe2}, or \cite[Theorem 5.3]{dgmo}. 

\begin{theorem}\label{thm:realsemsim}
  Let $\mc{O} = \wG\cdot p$ be the orbit of $p$.
  Write $H^1(\Gamma_p) = \{ [\gamma_1],\ldots, [\gamma_s]\}$.
  Suppose that for each  $\gamma_i$ there is
  $n_i\in Z^1 (N_{\wG}(\cF))$ with $\varphi(n_i)=\gamma_i$.
  Then $\mc{O}$ has a real point if and only if
  there exist $q\in \mc{O}\cap \cF$ and $i\in\{1,\ldots,s\}$ with
  $\bar q = \gamma_i^{-1} q$. If the latter holds,
  then $gq$ is a real point of $\mc{O}$, where $g\in \wG$ is such that
  $g^{-1} \bar g = n_i$. 
\end{theorem}

Using this theorem we classify the real semisimple elements, according to the complex family they lie in.

\textbf{First family} A brute force computer calculation shows that
$H^1 W_\gC$ is trivial. So by Theorem \ref{thm:realsemsim} it follows that
the orbits in this family having real points are 
exactly the orbits of $a_1u_1+a_2u_2+a_3u_3$ with all $a_i$ real. Let $Z$
denote the stabilizer of such a point in $\wG$. Since $Z$ is a group of order
$3^4$ it follows that $\Ho^1 Z$ is trivial (Criterion \ref{crt:pgroup}).
Therefore the complex orbit of $a_1u_1+a_2u_2+a_3u_3$ with all $a_i$ real
corresponds to exactly one real orbit. 

\textbf{Second family} Here $\Gamma_p$ has 18 elements. By Criterion
\ref{H:generators} $H^1 \Gamma_p$ consists of the trivial
class and the class of 
$$\gamma = \begin{pmatrix} 0 & -1 \\ -1 & 0 \end{pmatrix}.$$
So in this family there are two groups of orbits having real points. The first consists of the 
orbits of $a_1u_1+a_2u_2$ with $a_i$ real. Let $Z$ denote the stabilizer
of such a point, which is given explicitly in Section \ref{sec:semsimC}.
By Criterion \ref{crit:alg1}
have $H^1 Z = 1$. Therefore each complex orbit of a $a_1u_1+a_2u_2$ with $a_i$
real corresponds to exactly one real orbit. 

The second group consists of the orbits of $p$ with $\bar p = \gamma p$. These
are 
$p = a(u_1-u_2) + i b (u_1+u_2)$ with $a,b\in \R$. We have that
$\gamma$ is induced by 
\begin{equation}\label{eq:fam2n}
n = \begin{pmatrix} 0 & 0 & -1 \\ 0 & -1 & 0 \\ -1 & 0 & 0 \end{pmatrix} \times
\begin{pmatrix} -1 & 0 & 0 \\ 0 & 0 & -1 \\ 0 & -1 & 0 \end{pmatrix}\times
\begin{pmatrix} 0 & -1 & 0 \\ -1 & 0 & 0 \\ 0 & 0 & -1 \end{pmatrix},
\end{equation}
which is a cocycle. Setting
\begin{equation}\label{eq:fam2g}
g = \begin{pmatrix} -\tfrac{1}{2} & 0 & \tfrac{1}{2} \\ 0 & i & 0 \\ i & 0 & i\end{pmatrix}\times
\begin{pmatrix} -\tfrac{1}{2}i & 0 & 0\\ 0 & 1 & -1 \\ 0 & i & i \end{pmatrix}\times
\begin{pmatrix} 1 & -1 & 0 \\ i & i & 0 \\ 0 & 0 & -\tfrac{1}{2}i \end{pmatrix}
\end{equation}
we have $g^{-1}\bar g = n$. Let $v_1 = g\cdot (u_1-u_2)$ and $v_2 = g\cdot i(u_1+u_2)$. Then 
\begin{equation}\label{eq:newsemsim1}
\begin{aligned}  
v_1 &= -\mye{212}+\mye{200}+2\mye{120}-2\mye{111}+\tfrac{1}{2}\mye{022}-\tfrac{1}{2}\mye{001}\\
v_2 &= -\mye{222}-\mye{201}+2\mye{121}+2\mye{110}-\tfrac{1}{2}\mye{012}-\tfrac{1}{2}\mye{000}.
\end{aligned}
\end{equation}

So real representatives of these orbits are $a_1v_1+a_2v_2$, and the polynomial conditions translate 
to $a_1^2+a_2^2\neq 0$ and $a_2(a_2^2-3a_1^2)\neq 0$. Let $Z$ denote the stabilizer of such an 
element. A computer calculation shows
that $\Ho^1 Z^\circ=1$. Since the component group has order 9, Criterion
\ref{crit:alg1} shows that $\Ho^1 Z=1$. Hence the complex orbits with these
representatives correspond to one real orbit. 

\textbf{Third family} Here $\Gamma_p$ is generated by $-\zeta$. Hence $H^1 \Gamma_p = \{[1],[-1]\}$.
So also here the orbits with real representatives come in two groups. The first group has representatives $au_1$ with $a\in \R$. Let $Z$ be the stabilizer of such an element in $G$. Then 
$H^1 Z=1$, hence each of these orbits correspond to one real orbit. The second group has representatives $iau_1$ with $a\in \R$. We have that $-1$ is induced by the cocycle 
$n= A\times A\times A$ with 
\begin{equation}\label{eq:fam3n}
  A=\begin{pmatrix} 0 & -1 & 0 \\ -1 & 0 & 0 \\ 0 & 0 & -1 \end{pmatrix}.
\end{equation}  
Setting
\begin{equation}\label{eq:fam3g}
B = \begin{pmatrix} 1 & -1 & 0 \\ i & i & 0 \\ 0 & 0 & -\tfrac{1}{2}i \end{pmatrix}
\end{equation}
and $g=B\times B\times B$ we have that $g^{-1} \bar g = n$. 
We have that $g\cdot iu_1 = v$ with
\begin{equation}\label{eq:newsemsim2}
  v= -2\mye{001}-2\mye{010}-2\mye{100}+2\mye{111}-\tfrac{1}{8} \mye{222}.
\end{equation}  
Let $Z$ be the stabilizer of $v$ in $\wG$. By computing the Lie algebra of
$Z^\circ$, a further computer calculation shows that $\Ho^1 Z^\circ = 1$.
So again Criterion \ref{crit:alg1} shows that $\Ho^1 Z=1$.
Hence each complex orbit with representative
$av$, $a\in \R$, corresponds to one real orbit with the same representative. 

\textbf{Fourth family} Here $\Gamma_p$ is generated by $\zeta$. Hence
$H^1 \Gamma_p = 1$.
It follows that the orbits with real representatives are the orbits of $a(u_2-u_3)$ with 
$a\in \R$. Let $Z$ be the stabilizer of such an element in $\wG$. Criterion
\ref{crit:alg1} shows that $\Ho^1 Z=1$. Hence the complex orbit of
$a(u_2-u_1)$, with $a\in \R$, corresponds to one real orbit.

\section{Mixed Elements}\label{sec:mixed}

An $x\in \g_1$ is called {\em mixed} if it has Jordan decomposition $x = s+n$,
where $s$ and $n$ are two non zero semisimple and nilpotent elements,
respectively such that $[s,n]=0$.

\subsection{Mixed elements: the complex case}\label{mixed:complex}
Each mixed element is conjugate under $G_0$ to an element whose
semisimple part is one of the semisimple orbit representatives described in
Section \ref{semis}. So fix a semisimple part $s$ in one of the sets
$\gC_p^\circ$. Then consider its centralizer
$\a = \z_{\g}(s)$. Since $s\in \g_1$ this algebra inherits the grading
of $\g$: $\a = \a_{-1} \oplus \a_0 \oplus \a_1$ where $\a_1 =
\a\cap \g_i$. The possible nilpotent parts of mixed elements with
semisimple part $s$ lie in $\a_1$. The subalgebra $\a_0$ is the Lie algebra
of $Z_{G_0}(s)$. So $Z_{G_0}(s)$ acts on $\a_1$. It is clear that two mixed
elements $s+e_1$, $s+e_2$ are $G_0$-conjugate if and only if they are
$Z_{G_0}(s)$-conjugate. It follows that we can classify the possible nilpotent
parts of the mixed elements with semisimple part equal to $s$ by
classifying the nilpotent elements in $\a_1$ under the action of
$Z_{G_0}(s)$. The nilpotent orbits in $\a_1$ under the action of the
identity component $Z_{G_0}(s)^\circ$ can be classified with generic algorithms
for the classification of the nilpotent orbits of a $\theta$-group. Subsequently
it has to be seen which orbits are identified under by component group.
Also note that semisimple elements of a fixed $\gC_p^\circ$ have the same
centralizer $\g$ and the same stabilizer in $G_0$ by Lemma \ref{lem:stab}.
In other words, they do not depend on the semisimple element $s$, just on the
set  $\gC_p^\circ$.

We now give the classification of the possible nilpotent parts of mixed
elements. For each family of semisimple elements we have such a classification.
Here we omit the first and fifth families:
in the first case, the only possible nilpotent part is zero, in the latter
the semisimple part is zero.

We have computed the orbit classifications with the help of {\sf GAP}. For the
second and fourth families we obtained a result equivalent to the tables in
\cite{Nurmiev}. Therefore in these cases we have taken the same representatives.
For the third family our computation gave quite different results, showing 
that Table 2 in \cite{Nurmiev} is erroneous. 

Let $s$ be a semisimple element of
the $i$-th family. Let $\a$ be as above, and let $e\in \a_1$ be nilpotent.
Then $e$ lies in a homogeneous $\sll_2$-triple $(h,e,f)$ with $h\in \a_0$,
$f\in \a_{-1}$. For each nilpotent element $e$ in the classification we also
give a (not always very explicit) description of the stabilizer
$Z_{\wG}(s,h,e,f) = \{ g\in \wG \mid g\cdot s=s,\, g\cdot h=h,\,g\cdot e = e,\,
g\cdot f=f\}$ in terms of the identity component and its component group.

The possible nilpotent parts of a mixed element with semisimple part from
the second family are given in Table \ref{NilpotentMixed2}.

\begin{longtable}{|l|c|c|c|} 
%   \begin{center}
\caption{Nilpotent Part of Mixed Elements of Second Family}\label{NilpotentMixed2}\\
\hline
Nurmiev & Representative & $Z_{\wG}(s,h,e,f)^\circ$ & component group \\
\hline
1 & $\mye{021}+\mye{102}$ & $\id$ & $(\Ztre)^3$\\
    \hline
2 & $\mye{021}$ & $T_1$ & $(\Ztre)^2$\\
\hline
\end{longtable}

In the third family we have 
$s = a(\mye{000}+\mye{111}+\mye{222})$, where $a$ is a nonzero scalar.
Hence $s$ is stable under permutation of tensor factors. So if we let
$\a=\z_{\g}(s)$ then also $\a_1$ is stable under permutations of the tensor
factors. As seen in section
\ref{sec:perm} orbits that are obtained from each other by permutation
share many properties, such as the structure of the stabilizer.
In this case there are 16 orbits of possible nilpotent parts.
In the table below we give the representatives of these orbits in groups:
different members of the same group are related by a permutation and hence have
isomorphic centralizers. Therefore we only describe the centralizer for the
first element in a group. The numbering is according to the table for the
same orbits in Nurmiev's paper, \cite[Table 2]{Nurmiev}.

\begin{longtable}{|l|c|r|c|c|} 
%   \begin{center}
\caption{Nilpotent Part of Mixed Elements of Third Family}\label{NilpotentMixed3}\\
\hline
Nurmiev & Representative & $\pi$ & $Z_{\wG}(s,h,e,f)^\circ$ & component group\\
\hline
1 & $\mye{012}+\mye{021}+\mye{102}+\mye{120}$ &$\id$ & $\id$&$(\Ztre)^3$\\
& $\mye{102}+\mye{201}+\mye{012}+\mye{210}$ & $(1,2)$ & & \\
& $\mye{210}+\mye{120}+\mye{201}+\mye{021}$ & $(1,3)$ & &\\
    \hline
2 & $\mye{012}+\mye{021}+\mye{102}$ & $\id$&  $T_1$& $ (\Ztre)^2$\\
& $\mye{021}+\mye{012}+\mye{120}$ & $(2,3)$&& \\
& $\mye{102}+\mye{201}+\mye{012}$ & $(1,2)$&& \\
& $\mye{201}+\mye{102}+\mye{210}$ & $(1,2,3)$ &&\\
& $\mye{120}+\mye{210}+\mye{021}$ & $(1,3,2)$ &&\\
& $\mye{210}+\mye{120}+\mye{201}$ & $(1,3)$ && \\
\hline
4 & $\mye{012}+\mye{021}$ & $\id$ & $T_2 $&$\Ztre$\\
& $\mye{102}+\mye{201}$ & $(1,2)$ && \\
& $\mye{210}+\mye{120}$ & $(1,3)$ & &\\
\hline 
5 & $\mye{012}+\mye{120}$ &$\id$ & $T_2 $&$ \Ztre$ \\
& $\mye{102}+\mye{210}$ & $(1,2)$ & &\\
\hline
7 & $\mye{012}$ &$\id$ & $T_3$& $\id$\\
& $\mye{102}$ & $(1,2)$ & &\\
\hline
\end{longtable}

\begin{rmk}
\cite[Table 2]{Nurmiev} contains 8 orbits (apart from the zero orbit, which
we do not list here). His representative number 3 is conjugate to the second in
our list under 2. Nurmiev's  representative number 6 is conjugate to 
the second in our list under 5. Nurmiev's representative 8 is the second in
our list under 7.
\end{rmk}

\begin{longtable}{|l|c|c|c|} 
%   \begin{center}
\caption{Nilpotent Part of Mixed Elements of Fourth Family}\label{NilpotentMixed4}\\
\hline
Nurmiev & Representative & $Z_{\wG}(s,h,e,f)^\circ$ & component group\\
\hline
1 & $\mye{002}+\mye{020}+\mye{111}+\mye{200}$ & $\id $&$(\Ztre)^2$ \\
    \hline
2 & $\mye{002}+\mye{011}+\mye{020}+\mye{101}+\mye{110}+\mye{200}$ & $\id$&$(\Ztre)^2$\\
\hline
3 & $\mye{000}+\mye{111}$ & $\id$& $(\Ztre)^3 \ltimes \Z/2\Z$.\\
\hline 
4 & $\mye{001}+\mye{010}+\mye{100}$ & $T_1 $ &$\Ztre$ \\
\hline
5& $\mye{000}$ &   $\SL $ & $ (\Ztre)^2$\\
\hline
\end{longtable}

For later use we explicitly give the element of order 2 in the stabilizer
corresponding to $e=\e{000} + \e{111}$. It is   $A\times A\times A$ with 
$$A=\begin{pmatrix} 0&1&0\\1&0&0\\0&0&-1\end{pmatrix}.$$

\begin{rmk}\label{rem:Symnilp}
Here we comment on the classification up to permutation of the tensor factors
(see Section \ref{sec:perm}). The elements of Tables \ref{NilpotentMixed2},
\ref{NilpotentMixed4} are not conjugate to each other if we also consider
permutations of the tensor factors. So for the mixed elements of these families
there just remain the extra conjugacies coming from the semisimple parts
(Remark \ref{rem:Symsemsim}). For the nilpotent parts of elements whose
semisimple part is of the third family, the permutation action on the tensor
factors is given in detail in Table \ref{NilpotentMixed3}.
\end{rmk}

\subsection{Mixed elements: the real case}

The real semisimple elements are divided in two types. There are the real
semisimple elements in one of the families described in Section \ref{semis};
we say that those are canonical semisimple elements. Secondly we have the
non-canonical semisimple elements: they are $\wG$-conjugate, but not
$\wG(\R)$-conjugate to elements of one of the families. From Section
\ref{semsim:real} we see that they occur relative to the second and third
families. Correspondingly we have canonical and non-canonical mixed elements,
whose semisimple parts are canonical and non-canonical respectively.

Each complex orbit of a canonical mixed element has a real representative given
in the tables of the previous section. Let $u=s+e$ be such an element. Then
$e\in \a_1$, where $\a = \z_{\g}(s)$. Let $(h,e,f)$ be a homogeneous
$\sll_2$-triple in $\a$ containing $e$. Consider
$$Z_{\wG}(s,h,e,f) = \{ g\in \wG \mid g\cdot s=s,\, g\cdot h=h,\, g\cdot e=e,\,
g\cdot f = f\}.$$
Then the $\wG(\R)$-orbits contained in the complex orbit $\wG\cdot u$ correspond
bijectively to $\Ho^1 Z_{\wG}(s,h,e,f)$. From the descriptions of these groups in the tables
of Section \ref{mixed:complex} we see that this cohomology set is mostly
trivial. In those cases the complex orbit $\wG\cdot(s+e)$ contains one real
orbit $\wG(\R)\cdot(s+e)$. The one exception is the third element in
Table \ref{NilpotentMixed4}. In this case $|\Ho^1 Z_{\wG}(s,h,e,f)|=2$ and the
nontrivial cocycle is its element of order 2 which is $n=A\times A\times A$.
Let $g= B\times B\times B$ where
$$B=\begin{pmatrix} 1 & 1 & 0 \\ i & -i & 0 \\ 0 & 0 & \tfrac{1}{2} i
\end{pmatrix}.$$
Then $n=g^{-1}\bar g$ and $g\in Z_G(s)$. We have
$$g\cdot(\mye{000}+\mye{111}) = 2(\mye{000}-\mye{011}-\mye{101}-\mye{110}).$$
So here we get the extra real mixed orbits with representatives
\begin{equation}\label{eq:mixedcan}
  a(u_2-u_3) +2(\mye{000}-\mye{011}-\mye{101}-\mye{110}), \,a\in\R.
\end{equation}  

If $s$ is not canonical then $s$ is of the form $s = g\cdot s_0$, where
$s_0$ is a non-real element lying in one of the families. Furthermore,
$g^{-1}\bar g = n$, where $n$ is a cocycle in $Z_{\wG}(s_0)$ with $\bar s_0 =
n^{-1} s_0$.  If $e_0$ is a
nilpotent element commuting with $s_0$ from one of the tables of the
previous section, then $s+g\cdot e_0$ is a mixed element. However,
in general $g\cdot e_0$ is not real. So the first thing we have to do
is to see whether the orbit $Z_{\wG}(s)\cdot (g\cdot e_0)$ has real elements, and
find one in the affirmative case. On the other hand, if there are no real
points in this orbit, then this $e_0$ will not lead to mixed elements
with semisimple part equal to $s$, and can therefore be discarded.

For a semisimple element $p$ set $\u_p = \z_{\g}(p) \cap \g_1$.
Then $x\mapsto g\cdot x$ is a bijection $\u_{s_0} \to \u_s$. It also maps
$Z_{\wG}(s_0)$-orbits to $Z_{\wG}(s)$-orbits. Also for $x\in \u_{s_0}$ we have that
$n\bar x \in \u_{s_0}$ and $g\cdot x$ is real if and only if $n\bar x = x$.
So we define a map $\mu : \u_{s_0} \to \u_{s_0}$ by $\mu(x) = n\bar x$.
Now \cite[Lemma 5.3.1]{BorovoideGraafLe2} says the following.

\begin{lemma}\label{lem:Y}
Let $e\in \u_{s_0}$ be nilpotent and let $Y = Z_{\wG}(s_0)\cdot e$ be its orbit.
Let $y_0\in Y$. Then $\mu(Y) = Y$ if and only if $\mu(y_0)\in Y$.   
\end{lemma}

We use this lemma for classifying the non-canonical mixed orbits. From section
\ref{semsim:real} we see that this concerns only the second and third family.

\textbf{Second family:} Here $n$, $g$ are given in \eqref{eq:fam2n},
\eqref{eq:fam2g} respectively. We have $s_0 = a_1u_1+a_2u_2$ where the $a_2 =
-\bar a_1$. Furthermore, $s=b_1v_1+b_2v_2$ with $b_1,b_2\in \R$ and
$v_1,v_2$ given by \eqref{eq:newsemsim1}. 
From Table \ref{NilpotentMixed2} we see that there are two nilpotent orbits
in $\u_{s_0}$.

For the first orbit we have that $\mye{021}-\mye{210}$ is a
representative fixed under $\mu$. Furthermore,
$$g\cdot (\mye{021}-\mye{210}) = 2\mye{220}+2\mye{211}+\mye{021}-\mye{010}.$$
Let $e$ denote the latter element. Let $(h,e,f)$ be a homogeneous
$\sll_2$-triple containing $e$, then the stabilizer $Z_{\wG}(s,h,e,f)$ has
order $27$ and hence trivial Galois cohomology by Criterion \ref{crt:pgroup}.
So in this case we get one real orbit.

For the second orbit we have that $4i\mye{102}$ is a representative fixed
under $\mu$. Furthermore, $g\cdot 4i\mye{102} = \mye{102}$. Let $(h,e,f)$ be
as before and set $Z=Z_{\wG}(s,h,e,f)$. By computing the Lie algebra of $Z$ and
a small computer calculation, it is seen that $H^1 Z^\circ$ consists of two
elements. Since the component group is of order 9, Criterion \ref{crit:alg1}
implies that $H^1 Z = H^1 Z^\circ$. The nontrivial cocycle is
$$c = \diag(-1,1,-1)\times \diag(1,-1,-1)\times \diag(-1,-1,1).$$
So with
$$a= \diag(i,-1,i)\times \diag(-1,i,i)\times \diag(i,i,-1)$$ we get
$a^{-1}\bar a = c$. Furthermore, $a\cdot \mye{102} = -\mye{102}$.
Hence here we get two non-conjugate real nilpotent parts:
$\mye{102}$ and $-\mye{102}$.

In conclusion we have three nilpotent parts of a mixed element with semisimple
part equal to $g\cdot s_0$:

\begin{longtable}{|l|c|} 
\caption{Nilpotent parts of mixed elements with semisimple part $b_1v_1+b_2v_2$}
\label{tab:mixedR1}  
  \\
\hline
Nurmiev & representative \\
\hline
1 &  $2\mye{220}+2\mye{211}+\mye{021}-\mye{010}$\\
    \hline
2 & $\mye{102}$ \\
    & $-\mye{102}$\\
\hline
\end{longtable}

\textbf{Third family:} Here $n=A\times A\times A$ with $A$ as in
\eqref{eq:fam3n} and $g= B\times B\times B$ with $B$ as in \eqref{eq:fam3g}.
We have $s_0 = aiu_1$ where $a\in \R$ and $s=av$ with $v$ given by
\eqref{eq:newsemsim2}. From Table
\ref{NilpotentMixed3} we see that there are 16 nilpotent parts of mixed elements
with semisimple part $s_0$. They are divided into five groups with numbers
1, 2, 4, 5, 7. We deal with them in that order.

We have that $\mye{120}-\mye{102}-\mye{021}+\mye{012}$ is a representative of
orbit 1, fixed under $\mu$. It is left invariant by $g$. The stabilizer
in this case is of order 27, hence the Galois cohomology is trivial
(Criterion \ref{crt:pgroup}). 

Let $e= \mye{012}+\mye{021}+\mye{102}$. Then $e$ is a representative of the
first orbit in the second group. We have that $\mu(e)$ lies in the orbit of
$\mye{210}+\mye{120}+\mye{201}$ which is a representative of the sixth orbit
of the second group. Hence by Lemma \ref{lem:Y} the image of this orbit
in $\u_s$ has no real points. This then necessarily holds for all orbits of
the second group.

Let $e=\mye{210}-\mye{201}$. Then $e$ is a representative of the first orbit
in group number 4, fixed under $\mu$. Furthermore $g\cdot e = e$.
Let $(h,e,f)$ be a homogeneous $\sll_2$-triple containing $e$. Let $Z$
denote the stabilizer of the quadruple $(s,h,e,f)$ in $\wG$.
Then by a small computer calculation it can be shown $H^1 Z^\circ$ is
trivial. As the component group has order 3, it follows that $H^1 Z=1$.
So there is one real orbit in this case (Criterion  \ref{crit:alg1}).

Let $e=\mye{012}+\mye{120}$. Then $e$ is a representative of he first orbit in
group number 5. here $\mu(e)$ lies in the second orbit of group 5. Therefore,
by Lemma \ref{lem:Y} the image of this orbit in $\u_s$ has no real points.
The same holds for the second orbit of group number 5.

We have that $n\cdot \mye{012}$ lies in the orbit of $\mye{021}$. Hence again
by Lemma \ref{lem:Y} we see that the image of the first orbit of group number
7 has no real points.

Summarizing we get the following real mixed orbits:

\begin{longtable}{|l|c|} 
\caption{Nilpotent parts of mixed elements with semisimple part $av$}\label{tab:mixedR2} \\
\hline
Nurmiev & representative \\
\hline
1 &  $\mye{120}-\mye{102}-\mye{021}+\mye{012}$\\
  &  $\mye{210}-\mye{012}-\mye{201}+\mye{102}$\\
  &  $\mye{021}-\mye{201}-\mye{120}+\mye{210}$\\
    \hline
4 & $\mye{210}-\mye{201}$ \\
  & $\mye{120}-\mye{021}$\\
  & $\mye{012}-\mye{102}$\\
\hline
\end{longtable}

%\bibliographystyle{alpha}
%\bibliography{biblio.bib}

\begin{thebibliography}{MDDR{\etalchar{+}}04}

\bibitem[AMM97]{27}
Arvind, K~S Mallesh, and N~Mukunda.
\newblock A generalized pancharatnam geometric phase formula for three-level
  quantum systems.
\newblock {\em Journal of Physics A: Mathematical and General},
  30(7):2417--2431, apr 1997.

\bibitem[BD06]{30}
A~T Bölükbaşi and T~Dereli.
\newblock On the {SU}(3) parametrization of qutrits.
\newblock {\em Journal of Physics: Conference Series}, 36(1):28--32, apr 2006.

\bibitem[BDD{\etalchar{+}}09]{Oct}
L.~Borsten, D.~Dahanayake, M.~J. Duff, H.~Ebrahim, and W.~Rubens.
\newblock {Black Holes, Qubits and Octonions}.
\newblock {\em Phys. Rept.}, 471:113--219, 2009.

\bibitem[BDF{\etalchar{+}}12]{small-orbits}
L.~Borsten, M.~J. Duff, S.~Ferrara, A.~Marrani, and W.~Rubens.
\newblock {Small Orbits}.
\newblock {\em Phys. Rev. D}, 85:086002, 2012.

\bibitem[BdGL21]{BorovoideGraafLe}
Mikhail Borovoi, Willem~A. de~Graaf, and Hông~Vân Lê.
\newblock Real graded lie algebras, galois cohomology, and classification of
  trivectors in dimension 9 da aggiustare.
\newblock 2021.

\bibitem[BdGL22]{BorovoideGraafLe2}
Mikhail Borovoi, Willem~A. de~Graaf, and H\^{o}ng~V\^{a}n L\^{e}.
\newblock Classification of real trivectors in dimension nine.
\newblock {\em J. Algebra}, 603:118--163, 2022.

\bibitem[BDL12]{BH-2}
L.~Borsten, M.~J. Duff, and P.~Levay.
\newblock {The black-hole/qubit correspondence: an up-to-date review}.
\newblock {\em Class. Quant. Grav.}, 29:224008, 2012.

\bibitem[BDMR11]{BH-1}
L.~Borsten, M.~J. Duff, A.~Marrani, and W.~Rubens.
\newblock {On the Black-Hole/Qubit Correspondence}.
\newblock {\em Eur. Phys. J. Plus}, 126:37, 2011.

\bibitem[BRS{\etalchar{+}}21]{4-article}
M.~S. Blok, V.~V. Ramasesh, T.~Schuster, K.~O'Brien, J.~M. Kreikebaum,
  D.~Dahlen, A.~Morvan, B.~Yoshida, N.~Y. Yao, and I.~Siddiqi.
\newblock Quantum information scrambling on a superconducting qutrit processor.
\newblock {\em Phys. Rev. X}, 11:021010, Apr 2021.

\bibitem[Byr98]{3-wiki}
Mark Byrd.
\newblock Differential geometry on {SU}(3) with applications to three state
  systems.
\newblock {\em Journal of Mathematical Physics}, 39(11):6125--6136, nov 1998.

\bibitem[CFR01]{real-QM}
Carlton~M. Caves, Christopher~A. Fuchs, and Pranaw Rungta.
\newblock {Entanglement of Formation of an Arbitrary State of Two Rebits}.
\newblock {\em Found. Phys. Lett.}, 14(3):199--212, 2001.

\bibitem[Cho]{2-wiki}
C.Q. Choi.

\bibitem[CM00a]{2}
Carlton~M. Caves and Gerard~J. Milburn.
\newblock Qutrit entanglement.
\newblock {\em Optics Communications}, 179(1):439--446, 2000.

\bibitem[CM00b]{4-wiki}
Carlton~M. Caves and Gerard~J. Milburn.
\newblock Qutrit entanglement.
\newblock {\em Optics Communications}, 179(1-6):439--446, may 2000.

\bibitem[DdGMO22]{dgmo}
Heiko Dietrich, Willem~A. de~Graaf, Alessio Marrani, and Marcos Origlia.
\newblock Classification of four-rebit states.
\newblock {\em J. Geom. Phys.}, 179:Paper No. 104610, 31, 2022.

\bibitem[DF07]{17-rev}
M.~J. Duff and S.~Ferrara.
\newblock {E(6) and the bipartite entanglement of three qutrits}.
\newblock {\em Phys. Rev. D}, 76:124023, 2007.

\bibitem[dGGT19]{sla}
W.~A. de~Graaf and T.~GAP~Team.
\newblock {SLA}, computing with simple {L}ie algebras, {V}ersion 1.5.3.
\newblock {\texttt{https://gap-packages.github.io/}\discretionary
  {}{}{}\texttt{sla/}}, Nov 2019.
\newblock Refereed GAP package.

\bibitem[dGM20]{dGM1}
Willem~A. de~Graaf and Alessio Marrani.
\newblock {Real forms of embeddings of maximal reductive subalgebras of the
  complex simple Lie algebras of rank up to 8}.
\newblock {\em J. Phys. A}, 53(15):155203, 2020.

\bibitem[DGPS22]{DGPS}
Wolfram Decker, Gert-Martin Greuel, Gerhard Pfister, and Hans Sch\"onemann.
\newblock {\sc Singular} {4-3-0} --- {A} computer algebra system for polynomial
  computations.
\newblock \verb+http://www.singular.uni-kl.de+, 2022.

\bibitem[DHW08]{Dasgupta-Wissanji}
Keshav Dasgupta, Veronique Hussin, and Alisha Wissanji.
\newblock {Quaternionic Kahler Manifolds, Constrained Instantons and the Magic
  Square. I.}
\newblock {\em Nucl. Phys. B}, 793:34--82, 2008.

\bibitem[FG98]{FG1}
Sergio Ferrara and Murat Gunaydin.
\newblock {Orbits of exceptional groups, duality and BPS states in string
  theory}.
\newblock {\em Int. J. Mod. Phys. A}, 13:2075--2088, 1998.

\bibitem[FG06]{FG2}
Sergio Ferrara and Murat Gunaydin.
\newblock {Orbits and Attractors for N=2 Maxwell-Einstein Supergravity Theories
  in Five Dimensions}.
\newblock {\em Nucl. Phys. B}, 759:1--19, 2006.

\bibitem[FGK06]{FGimK}
Sergio Ferrara, Eric~G. Gimon, and Renata Kallosh.
\newblock {Magic supergravities, N= 8 and black hole composites}.
\newblock {\em Phys. Rev. D}, 74:125018, 2006.

\bibitem[FMO{\etalchar{+}}11]{FMOSY}
Sergio Ferrara, Alessio Marrani, Emanuele Orazi, Raymond Stora, and Armen
  Yeranyan.
\newblock {Two-Center Black Holes Duality-Invariants for stu Model and its
  lower-rank Descendants}.
\newblock {\em J. Math. Phys.}, 52:062302, 2011.

\bibitem[FMT12]{SE}
Sergio Ferrara, Alessio Marrani, and Mario Trigiante.
\newblock {Super-Ehlers in any Dimension}.
\newblock {\em JHEP}, 11:068, 2012.

\bibitem[FMZ13]{JP}
Sergio Ferrara, Alessio Marrani, and Bruno Zumino.
\newblock {Jordan Pairs, E6 and U-Duality in Five Dimensions}.
\newblock {\em J. Phys. A}, 46:065402, 2013.

\bibitem[Gam16]{24}
Omar Gamel.
\newblock Entangled bloch spheres: Bloch matrix and two-qubit state space.
\newblock {\em Phys. Rev. A}, 93:062320, Jun 2016.

\bibitem[GAP22]{GAP4}
The GAP~Group.
\newblock {\em {GAP -- Groups, Algorithms, and Programming, Version 4.12.2}},
  2022.

\bibitem[GBD{\etalchar{+}}20]{36}
Pranav Gokhale, Jonathan~M. Baker, Casey Duckering, Frederic~T. Chong,
  Natalie~C. Brown, and Kenneth~R. Brown.
\newblock Extending the frontier of quantum computers with qutrits.
\newblock {\em IEEE Micro}, 40(3):64--72, 2020.

\bibitem[GM21]{Gharahi-Mancini}
Masoud Gharahi and Stefano Mancini.
\newblock Algebraic-geometric characterization of tripartite entanglement.
\newblock {\em Phys. Rev. A}, 104:042402, Oct 2021.

\bibitem[GOV94]{gov}
V.~V. Gorbatsevich, A.~L. Onishchik, and {\`E}.~B. Vinberg.
\newblock {\em Structure of {L}ie Groups and {L}ie Algebras}, volume~41 of {\em
  Encyclopaedia of Mathematical Sciences}.
\newblock Springer-Verlag, Berlin, 1994.
\newblock In: {L}ie {G}roups and {L}ie Algebras III, A. L. Onishchik and {\`E}.
  B. Vinberg (eds.).

\bibitem[Gra17]{graaf17}
Willem A.~de Graaf.
\newblock {\em Computation with linear algebraic groups}.
\newblock Monographs and Research Notes in Mathematics. CRC Press, Boca Raton,
  FL, 2017.

\bibitem[GST83]{GST1}
M.~Gunaydin, G.~Sierra, and P.~K. Townsend.
\newblock {Exceptional Supergravity Theories and the MAGIC Square}.
\newblock {\em Phys. Lett. B}, 133:72--76, 1983.

\bibitem[GST84]{GST2}
M.~Gunaydin, G.~Sierra, and P.~K. Townsend.
\newblock {The Geometry of N=2 Maxwell-Einstein Supergravity and Jordan
  Algebras}.
\newblock {\em Nucl. Phys. B}, 242:244--268, 1984.

\bibitem[Hel78]{helgason}
Sigurdur Helgason.
\newblock {\em Differential geometry, {L}ie groups, and symmetric spaces},
  volume~80 of {\em Pure and Applied Mathematics}.
\newblock Academic Press Inc. [Harcourt Brace Jovanovich Publishers], New York,
  1978.

\bibitem[HFL05]{9}
César Herreño-Fierro and J.~R. Luthra.
\newblock Generalized concurrence and limits of separability for two qutrits.
\newblock 2005.

\bibitem[HHDR21]{3-article}
Alexander~D. Hill, Mark~J. Hodson, Nicolas Didier, and Matthew~J. Reagor.
\newblock Realization of arbitrary doubly-controlled quantum phase gates.
\newblock 2021.

\bibitem[HT95]{HT}
C.~M. Hull and P.~K. Townsend.
\newblock {Unity of superstring dualities}.
\newblock {\em Nucl. Phys. B}, 438:109--137, 1995.

\bibitem[Hur63]{Hurwitz}
Adolf Hurwitz.
\newblock {\"U}ber die komposition der quadratischen formen von beliebig vielen
  variablen.
\newblock pages 565--571, 1963.

\bibitem[JCB22]{Neutr-Qutrits}
Abhishek~Kumar Jha, Akshay Chatla, and Bindu~A. Bambah.
\newblock Neutrinos as qubits and qutrits.
\newblock 2022.

\bibitem[Kac80]{Kac-80}
V.G. Kac.
\newblock Some remarks on nilpotent orbits.
\newblock {\em Journal of Algebra}, 64(1):190--213, 1980.

\bibitem[Kac90]{kac}
V.~G. Kac.
\newblock {\em Infinite Dimensional Lie Algebras}.
\newblock Cambridge University Press, Cambridge, third edition, 1990.

\bibitem[KGRS03]{7}
A.~B. Klimov, R.~Guzm\'an, J.~C. Retamal, and C.~Saavedra.
\newblock Qutrit quantum computer with trapped ions.
\newblock {\em Phys. Rev. A}, 67:062313, Jun 2003.

\bibitem[KMSM97]{29}
G.~Khanna, S.~Mukhopadhyay, R.~Simon, and N.~Mukunda.
\newblock Geometric phases forsu(3) representations and three level quantum
  systems.
\newblock {\em Annals of Physics}, 253(1):55--82, 1997.

\bibitem[LT09]{lehta}
Gustav~I. Lehrer and Donald~E. Taylor.
\newblock {\em Unitary reflection groups}, volume~20 of {\em Australian
  Mathematical Society Lecture Series}.
\newblock Cambridge University Press, Cambridge, 2009.

\bibitem[Luc21]{Lucietti-multi}
James Lucietti.
\newblock {All Higher-Dimensional Majumdar\textendash{}Papapetrou Black Holes}.
\newblock {\em Annales Henri Poincare}, 22(7):2437--2450, 2021.

\bibitem[LWL{\etalchar{+}}08]{6-wiki}
B.~P. Lanyon, T.~J. Weinhold, N.~K. Langford, J.~L. O'Brien, K.~J. Resch,
  A.~Gilchrist, and A.~G. White.
\newblock Manipulating biphotonic qutrits.
\newblock {\em Phys. Rev. Lett.}, 100:060504, Feb 2008.

\bibitem[LYF13]{37}
Bin Li, Zu-Huan Yu, and Shao-Ming Fei.
\newblock Geometry of quantum computation with qutrits.
\newblock {\em Scientific Reports}, 3(1), sep 2013.

\bibitem[MDDR{\etalchar{+}}04]{5-wiki}
A.~Melikidze, V.~V. Dobrovitski, H.~A. De~Raedt, M.~I. Katsnelson, and B.~N.
  Harmon.
\newblock Parity effects in spin decoherence.
\newblock {\em Phys. Rev. B}, 70:014435, Jul 2004.

\bibitem[MM97]{28}
K~S Mallesh and N.~Mukunda.
\newblock The algebra and geometry ofsu(3) matrices.
\newblock {\em Pramana}, 49:371--383, 1997.

\bibitem[MMS99]{86-rev}
Juan~Martin Maldacena, Gregory~W. Moore, and Andrew Strominger.
\newblock {Counting BPS black holes in toroidal Type II string theory}.
\newblock 3 1999.

\bibitem[MOPL{\etalchar{+}}09]{2-article}
Fran{\c{c}}ois Mallet, Florian~R. Ong, Agustin Palacios-Laloy, Fran{\c{c}}ois
  Nguyen, Patrice Bertet, Denis Vion, and Daniel Esteve.
\newblock Single-shot qubit readout in circuit quantum electrodynamics.
\newblock {\em Nature Physics}, 5(11):791--795, sep 2009.

\bibitem[MP86]{T2}
Robert~C. Myers and M.~J. Perry.
\newblock {Black Holes in Higher Dimensional Space-Times}.
\newblock {\em Annals Phys.}, 172:304, 1986.

\bibitem[MPRR17]{MRR1}
Alessio Marrani, Gianfranco Pradisi, Fabio Riccioni, and Luca Romano.
\newblock {Nonsupersymmetric magic theories and Ehlers truncations}.
\newblock {\em Int. J. Mod. Phys. A}, 32(19n20):1750120, 2017.

\bibitem[MR19]{MRR2}
Alessio Marrani and Luca Romano.
\newblock {Orbits in nonsupersymmetric magic theories}.
\newblock {\em Int. J. Mod. Phys. A}, 34(32):1950190, 2019.

\bibitem[NJDH{\etalchar{+}}13]{1-wiki}
Peter B~R Nisbet-Jones, Jerome Dilley, Annemarie Holleczek, Oliver Barter, and
  Axel Kuhn.
\newblock Photonic qubits, qutrits and ququads accurately prepared and
  delivered on demand.
\newblock {\em New Journal of Physics}, 15(5):053007, may 2013.

\bibitem[Nur00a]{Nur00}
A~G Nurmiev.
\newblock Orbits and invariants of cubic matrices of order three.
\newblock {\em Sbornik: Mathematics}, 191(5):717, jun 2000.

\bibitem[Nur00b]{Nurmiev}
AG~Nurmiev.
\newblock Orbits and invariants of third-order matrices.
\newblock {\em Mat. Sb}, 191(5):101--108, 2000.

\bibitem[PV08]{Pal-Vertesi}
K\'aroly~F. P\'al and Tam\'as V\'ertesi.
\newblock Efficiency of higher-dimensional hilbert spaces for the violation of
  bell inequalities.
\newblock {\em Phys. Rev. A}, 77:042105, Apr 2008.

\bibitem[Ser97]{serre}
J.-P. Serre.
\newblock {\em Galois Cohomology}.
\newblock Springer-Verlag, Berlin, 1997.

\bibitem[Ste64]{steinberg}
Robert Steinberg.
\newblock Differential equations invariant under finite reflection groups.
\newblock {\em Trans. Amer. Math. Soc.}, 112:392--400, 1964.

\bibitem[Ste75]{steinberg2}
Robert Steinberg.
\newblock Torsion in reductive groups.
\newblock {\em Adv. Math.}, 15:63--92, 1975.

\bibitem[Stu60]{Stueck}
E.C.G. Stueckelberg.
\newblock Quantum theory in real hilbert space.
\newblock {\em Helv. Phys. Acta}, 33:727, 1960.

\bibitem[Tan63]{T1}
F.~R. Tangherlini.
\newblock {Schwarzschild field in n dimensions and the dimensionality of space
  problem}.
\newblock {\em Nuovo Cim.}, 27:636--651, 1963.

\bibitem[Tay12]{taylor}
D.~E. Taylor.
\newblock Reflection subgroups of finite complex reflection groups.
\newblock {\em J. Algebra}, 366:218--234, 2012.

\bibitem[TJD11]{8}
Frank Tabakin and Bruno Juliá-Díaz.
\newblock Qcwave – a mathematica quantum computer simulation update.
\newblock {\em Computer Physics Communications}, 182(8):1693--1707, 2011.

\bibitem[VE88]{VinEl}
E.~B. Vinberg and A.~G. Elashvili.
\newblock A classification of the trivectors of nine- dimensional space
  (russian.
\newblock {\em Trudy Sem. Vektor. Tenzor. Anal. 18 (1978)}, English
  translation: Selecta Math. Sov., 7(5):63–--98, 1988.

\bibitem[Vin76a]{vinberg}
{\`E}.~B. Vinberg.
\newblock The {W}eyl group of a graded {L}ie algebra.
\newblock {\em Izv. Akad. Nauk SSSR Ser. Mat.}, 40(3):488--526, 1976.
\newblock English translation: Math. USSR-Izv. 10, 463--495 (1976).

\bibitem[Vin76b]{Vinberg-Weyl}
È~B Vinberg.
\newblock The weyl group of a graded lie algebra.
\newblock {\em Mathematics of the USSR-Izvestiya}, 10(3):463, jun 1976.

\bibitem[Vin79]{vinberg2}
{\`E}.~B. Vinberg.
\newblock Classification of homogeneous nilpotent elements of a semisimple
  graded {L}ie algebra.
\newblock {\em Trudy Sem. Vektor. Tenzor. Anal.}, (19):155--177, 1979.
\newblock English translation: Selecta Math. Sov. 6, 15-35 (1987).

\bibitem[Wal17]{wallach}
Nolan~R. Wallach.
\newblock {\em Geometric invariant theory}.
\newblock Universitext. Springer, Cham, 2017.
\newblock Over the real and complex numbers.

\bibitem[Wan07]{Dray-or-Wangberg}
Aaron Wangberg.
\newblock {The Structure of E(6)}.
\newblock 8 2007.

\end{thebibliography}

\newcommand{\etalchar}[1]{$^{#1}$}

\end{document}